\RequirePackage{fix-cm}

\documentclass[20pt,a4paper]{svjour3}       
\usepackage[table, svgnames, dvipsnames]{xcolor}
\usepackage{tabularx, ragged2e, booktabs, caption}
\usepackage{amssymb}
\usepackage{amsmath}
\usepackage{lipsum,parskip,booktabs,graphicx,tabulary,tabularx}
\usepackage{color,soul}
\usepackage{geometry}
\smartqed  
\usepackage{tabularx,ragged2e,booktabs,caption}
\newcolumntype{C}[1]{>{\Centering}m{#1}}

  \usepackage{hyperref}
\usepackage{graphicx}

\usepackage{pdflscape}
\begin{document}

\title{Quantum Codes from additive constacyclic codes over a mixed alphabet and the MacWilliams identities}


\author{Indibar Debnath \and Ashutosh Singh \and Om Prakash$^{*}$ \and Abdollah Alhevaz 
}


\institute{${*}$
\at
              Department of Mathematics\\
              Indian Institute of Technology Patna, Bihar 801 106, India \\
              \email{indibar$\_$1921ma07@iitp.ac.in, ashutosh$\_$1921ma05@iitp.ac.in, om@iitp.ac.in(*corresponding author)}\\
              \and Abdollah Alhevaz
            \at
             Faculty of Mathematical Sciences\\ Shahrood University of Technology\\P.O. Box: 316-3619995161, Shahrood, Iran\\
              \email{a.alhevaz@gmail.com}
}
\date{Received: date / Accepted: date}
\maketitle

\begin{abstract}
Let $\mathbb{Z}_p$ be the ring of integers modulo a prime number $p$ where $p-1$ is a quadratic residue modulo $p$. This paper presents the study of constacyclic codes over chain rings $\mathcal{R}=\frac{\mathbb{Z}_p[u]}{\langle u^2\rangle}$ and $\mathcal{S}=\frac{\mathbb{Z}_p[u]}{\langle u^3\rangle}$. We also study additive constacyclic codes over $\mathcal{R}\mathcal{S}$ and $\mathbb{Z}_p\mathcal{R}\mathcal{S}$ using the generator polynomials over the rings $\mathcal{R}$ and $\mathcal{S},$ respectively. Further, by defining Gray maps on $\mathcal{R}$, $\mathcal{S}$ and $\mathbb{Z}_p\mathcal{R}\mathcal{S},$ we obtain some results on the Gray images of additive codes. Then we give the weight enumeration and MacWilliams identities corresponding to the additive codes over $\mathbb{Z}_p\mathcal{R}\mathcal{S}$. Finally, as an application of the obtained codes, we give quantum codes using the CSS construction.

\keywords{Constacyclic code, chain ring, Frobenius ring, Gray image, MacWilliams identity, Quantum code}
\end{abstract}

\section{Introduction}
Initially, codes were studied over finite fields. From the $1970$s onward, the study of codes over rings started; see \cite{Blake1,Blake2}. But this study over rings found momentum and created a lot of interest among the researchers after the extraordinary work by Hammons et al. \cite{Hammons} in 1993. Recently noncommutative rings have been considered in many works to study and obtain better codes \cite{Habibul1,ShikhaP1,ShikhaP2}. Still, the study has been done mostly on commutative rings for ease of computation. In $1997$, Rif\`a and Pujol \cite{Rifa} first encountered codes over mixed alphabets. Later, Borges et al. studied $\mathbb{Z}_2\mathbb{Z}_4$-additive codes and defined the duality of such codes. Abualrub et al. \cite{Abualrub3} studied algebraic structures of $\mathbb{Z}_2\mathbb{Z}_4$-additive cyclic codes and determined a set of generator polynomials of those codes. They showed that the duals of $\mathbb{Z}_2\mathbb{Z}_4$-additive cyclic codes are also cyclic and further obtained some optimal codes from the $\mathbb{Z}_2\mathbb{Z}_4$-additive cyclic codes. As a natural generalization of $\mathbb{Z}_2\mathbb{Z}_4$-additive codes, Aydogdu and Siap \cite{Aydogdu1} in $2013$ investigated the algebraic structure of $\mathbb{Z}_2\mathbb{Z}_{2^s}$-additive codes and presented the standard form of the generator and parity-check matrices. Aydogdu et al. \cite{Aydogdu3} determined the algebraic structure of linear and cyclic codes over $\mathbb{Z}_2\mathbb{Z}_2[u^3]$. Islam et al. \cite{Habibul3} studied cyclic and constacyclic codes over $\mathbb{Z}_4\mathbb{Z}_4[u]$ and determined their generator polynomials and the minimal spanning sets. They defined new Gray maps and showed that the images of the $\mathbb{Z}_4\mathbb{Z}_4[u]$-additive constacyclic codes and skew $\mathbb{Z}_4\mathbb{Z}_4[u]$-additive constacyclic codes are generalized quasi-cyclic codes over $\mathbb{Z}_4$. The structure of $\mathbb{Z}_4\mathbb{Z}_4[u^3]$ has been considered by Prakash et al. \cite{Sir} to study additive cyclic and constacyclic codes. First, they obtained the generator polynomials along with the minimal generating set of additive cyclic codes and then extended the results to determine the structure of additive constacyclic codes. Recently, in $2022$, a book has been published where Borges et al. \cite{Borges} have thoroughly studied the $\mathbb{Z}_2\mathbb{Z}_4$-linear codes. They have investigated various properties to show the importance of these codes. Further, they studied the dual structure, the rank and kernel, and encoding and decoding of $\mathbb{Z}_2\mathbb{Z}_4$-linear codes.

One of the most important and frequently studied linear codes is the class of constacyclic codes. This family of linear codes has a wide range of applications in information technology. Due to its beautiful algebraic structure, constacyclic codes are easy to implement and shift registers can encode them. Over the years, researchers have considered constacyclic codes to study different aspects of coding theory extensively, see \cite{Aydin,Chen,Indibar,Habibul2,Rk2}.

Weight distribution in coding theory is another important aspect. The weight enumerator of a linear code of length $n$ indicates the number of codewords of each possible weight $0,1,\ldots,n$. In $1963$, a remarkable work of MacWilliams \cite{MacWilliams} proposed a formula that relates the weight enumerator of a code with that of its dual. Yildiz and Karadeniz \cite{Yildiz} considered linear codes over $\mathbb{Z}_4 + u\mathbb{Z}_4$ and proved the MacWilliams identities for complete, symmetrized, and Lee weight enumerators. Aydogdu et al. \cite{Aydogdu2} introduced a new class of additive codes, $\mathbb{Z}_2\mathbb{Z}_2[u]$-additive codes and they proved a MacWilliams-type identity on the weight enumerator of such codes. Later, Tang et al. \cite{Tang} studied the existence of the MacWilliams-type identities for the Lee and Euclidean weight enumerators and provided necessary and sufficient conditions for the existence of those identities over $\mathbb{Z}_l$. Recently, in $2021$, Bhaintwal and Biswas \cite{Bhaintwal} studied the algebraic structure of $\mathbb{Z}_p\mathbb{Z}_p[u]/\langle u^k\rangle$-cyclic codes and established the MacWilliams identities for complete weight enumerators of $\mathbb{Z}_p\mathbb{Z}_p[u]/\langle u^k\rangle$-linear codes.

The above research works motivate us to study $\mathbb{Z}_p\mathcal{R}\mathcal{S}$-additive codes where $\mathcal{R}=\frac{\mathbb{Z}_p[u]}{\langle u^2\rangle}$, $\mathcal{S}=\frac{\mathbb{Z}_p[u]}{\langle u^3\rangle}$ and $\mathbb{Z}_p$ is the ring of integers modulo a prime $p$, where $p-1$ is a quadratic residue modulo $p$. The ring $\mathbb{Z}_p\mathcal{R}\mathcal{S}$ is a generalization of the rings such as $\mathbb{Z}_2\mathbb{Z}_2[u]$, $\mathbb{Z}_2\mathbb{Z}_2[u^3]$ and $\mathbb{Z}_p\mathbb{Z}_p[u]/\langle u^k\rangle$, on which additive codes have already been studied. But the study of additive constacyclic codes over the ring $\mathbb{Z}_p\mathcal{R}\mathcal{S}$ is not yet available in the literature. Specifically, we study the $\mathbb{Z}_p\mathcal{R}\mathcal{S}$-additive constacyclic codes of block length $(q,r,s)$ and derive the form of generators of these codes. First, we derive the form of generators of additive constacyclic codes over $\mathcal{R}$ and $\mathcal{S}$ each. Then using these generators, we find the generator of $\mathbb{Z}_p\mathcal{R}\mathcal{S}$-additive constacyclic codes. We define a suitable inner product on both $\mathcal{R}^r\mathcal{S}^s$ and $\mathbb{Z}_p^q\mathcal{R}^r\mathcal{S}^s$, and use them to find a relation between an additive constacyclic code and its dual. Next, we define Gray maps on $\mathcal{R}^r$, $\mathcal{S}^s$ and $\mathbb{Z}_p^q\mathcal{R}^r\mathcal{S}^s$, respectively, and study the Gray images of additive cyclic codes and additive constacyclic codes over $\mathcal{R}$, $\mathcal{S}$ and $\mathbb{Z}_p\mathcal{R}\mathcal{S}$. Then we obtain the MacWilliams identities of $\mathbb{Z}_p\mathcal{R}\mathcal{S}$-additive codes corresponding to the complete, Hamming, Symmetrized, and Lee weight enumerators. Finally, as an application, we find some quantum codes over $\mathcal{R}$.

This paper is arranged as follows: In section $2$, we recall some basic definitions and results, which will be required later. Section $3$ deals mainly with the generators of additive constacyclic codes over $\mathcal{R}$, $\mathcal{S}$ and $\mathbb{Z}_p\mathcal{R}\mathcal{S}$ each. Section $4$ revolves around the Gray maps defined over $\mathcal{R}^r$, $\mathcal{S}^s$ and $\mathbb{Z}_p^q\mathcal{R}^r\mathcal{S}^s$, and also their images. In section $5$, we obtain the MacWilliams identities of $\mathbb{Z}_p\mathcal{R}\mathcal{S}$-additive codes over several weight enumerators. We find some quantum codes over $\mathcal{R}$ in section $6$. Section $7$ concludes our work.

\section{Preliminaries}\label{sec 2}
In this section, we state some basic definitions and results. Let $\mathcal{R}=\frac{\mathbb{Z}_p[u]}{\langle u^2\rangle}$ and $\mathcal{S}=\frac{\mathbb{Z}_p[u]}{\langle u^3\rangle}$ where $\mathbb{Z}_p$ is a ring of integers modulo a prime $p$ and $p-1$ being a quadratic residue in $\mathbb{Z}_p$. Note that $\mathcal{R}$ and $\mathcal{S}$ are chain rings of order $p^2$ and $p^3$, respectively.\vspace{0.01 cm}

For any two positive integers $r$ and $s$, $\mathcal{R}^r\times \mathcal{S}^s$ is an additive group and $\mathcal{R}^r\times \mathcal{S}^s$ forms an $\mathcal{S}$-module with the scalar multiplication
 \begin{align*}
 d\cdot (x,y) &= d\cdot (x_0,x_1, \ldots, x_{r-1},y_0,y_1, \ldots, y_{s-1})\\
 &= (d^{\prime\prime}x_0,d^{\prime\prime}x_1, \ldots,d^{\prime\prime} x_{r-1},dy_0,dy_1, \ldots, dy_{s-1}),
 \end{align*}
 where $d\in\mathcal{S},~ x\in\mathcal{R}^r,~ y\in\mathcal{S}^s,~ d^{\prime\prime} = d(\text{mod } u^2)$.\vspace{0.01 cm}

 Similarly, for any three positive integers $q,r\text{ and }s$, $\mathbb{Z}_p^q\times\mathcal{R}^r\times\mathcal{S}^s$ also forms a module over $\mathcal{S}$ with the scalar multiplication
 \begin{align*}
 d\cdot (w,x,y) &= d\cdot (w_0,w_1,\ldots,w_{q-1},x_0,x_1, \ldots, x_{r-1},y_0,y_1, \ldots, y_{q-1})\\
 &= (d^{\prime}w_0,d^{\prime}w_1,\ldots,d^{\prime}w_{q-1},d^{\prime\prime}x_0,d^{\prime\prime}x_1, \ldots,d^{\prime\prime} x_{r-1},dy_0,dy_1, \ldots, dy_{q-1}),
 \end{align*}
 where $d\in\mathcal{S},~ w\in\mathbb{Z}_p^q,~ x\in\mathcal{R}^r,~ y\in\mathcal{S}^s,~d^{\prime} = d~(\text{mod } u) \text{ and }d^{\prime\prime} = d~(\text{mod } u^2)$.\vspace{0.01 cm}

The direct product of $\mathcal{R}$ and $\mathcal{S}$ is denoted by $\mathcal{R}\mathcal{S}$ and the direct product of $\mathbb{Z}_p$, $\mathcal{R}$ and $\mathcal{S}$ is denoted by $\mathbb{Z}_p\mathcal{R}\mathcal{S}$. Now, we define an $\mathcal{R}\mathcal{S}$-additive code and a $\mathbb{Z}_p\mathcal{R}\mathcal{S}$-additive code.

\begin{definition}
An $\mathcal{R}\mathcal{S}$-additive code of block length $(r,s)$ is an $\mathcal{S}$-submodule of $\mathcal{R}^r\mathcal{S}^s$ and similarly, a $\mathbb{Z}_p\mathcal{R}\mathcal{S}$-additive code of block length $(q,r,s)$ is an $\mathcal{S}$-submodule of $\mathbb{Z}_p^q\mathcal{R}^r\mathcal{S}^s$.
\end{definition}

\begin{example}\label{Example1}
 Consider the code $\mathfrak{C}_1 = \langle\{ (1,0;0,u;1+u^2,0),(0,1;1+u,0;0,1+u) \}\rangle$ of block length $(2,2,2)$ over $\mathbb{Z}_2\mathcal{R}\mathcal{S}$. Then
 \begin{align*}
    \mathfrak{C}_1 =&~ \langle\{(1,0;0,u;1+u^2,0),(0,0;0,0;u,0),(0,0;0,0;u^2,0),(0,1;1+u,0;0,1+u),\\
 &~~~ (0,0;u,0;0,u+u^2),(0,0;0,0;0,u^2)\}\rangle
 \end{align*}
over $\mathbb{Z}_2$, i.e., $\mathfrak{C}_1$ is a vector space over $\mathbb{Z}_2$ of dimension $6$.
\end{example}
Next, we recall the definition of a Frobenius ring.

\begin{definition}\cite{Dougherty}
 Let $R$ be a ring with unity. Then $R$ is called a Frobenius ring if $R$ is Artinian and $R/\text{Rad}(R)\cong\text{Soc}(R)$ both (left as well as right) $R$-modules. $\text{Rad}(R)$ denotes the Jacobson radical of $R$ and $\text{Soc}(R)$ denotes the socle of $R$ as an $R$-module.
\end{definition}

\begin{lemma}
The rings $\mathcal{R},~ \mathcal{S}~ \text{and}~ \mathbb{Z}_p\mathcal{R}\mathcal{S}$ are Frobenius.
\end{lemma}

\begin{definition}\cite{Claasen}
 A character $\chi$ of a ring $R$ is a group homomorphism from $R$ to $\mathbb{C}^{\ast}$ where $\mathbb{C}^{\ast}$ is the group of all nonzero complex numbers.
\end{definition}

\begin{definition}\cite{Wood}
 Let $R$ be a finite ring and let $\widehat{R}$ be the set of all characters of $R$. If there exists an $R$-module isomorphism $f: R\rightarrow\widehat{R},$ then $\chi = f(1)$ is said to be a generating character of $R$.
\end{definition}

\begin{lemma}\label{Lemma5}\cite{Claasen}
Let $\chi$ be a character of a finite ring $R$. Then $\chi$ is a generating character if and only if $\text{ker}(\chi)$ contains no nonzero ideals of $R$.
\end{lemma}
Now, let us recall the definitions of a few special classes of codes.

\begin{definition}
Let $R$ be a ring, $n$ be a positive integer, and $\lambda$ be a unit in $R$. We denote the $\lambda$-constacyclic shift operator by $\sigma_{\lambda}$ and it is defined on $R^n$ by$$\sigma_{\lambda}(x_0, x_1,\ldots, x_{n-1}) = (\lambda x_{n-1},x_0,x_1,\ldots, x_{n-2}),$$where $x_i\in R$ for $i = 0,1,\ldots,n-1$. An $R$-additive code $\mathcal{C}$ of length $n$ is said to be a $\lambda$-constacyclic code if $\mathcal{C}$ is invariant under the map $\sigma_{\lambda}$. In particular, when $\lambda = 1$, we denote the operator $\sigma_1$ simply as $\sigma$, which is known as the cyclic shift operator. An $R$-additive code invariant under $\sigma$ is called a cyclic code.
\end{definition}

\begin{definition}\cite{Conan}
Let $R$ be a ring and $n = lm$ where $l,m$ are positive integers. We denote the $l$-quasi-cyclic shift operator by $\theta_l$ and it is defined on $R^n$ by $$\theta_l(x_0\vert x_1\vert\ldots\vert x_{l-1}) = (\sigma(x_0)\vert\sigma(x_1)\vert\ldots\vert\sigma(x_{l-1}))$$where $x_i\in R^m$ for $i = 0,1,\ldots,l-1$ and $\sigma$ is the cyclic shift operator on $R^m$. An $R$-additive code $\mathcal{C}$ of length $n$ is said to be an $l$-quasi-cyclic code if $\mathcal{C}$ is invariant under the map $\theta_l$.
\end{definition}

\begin{definition}\cite{Hill}
Let $R$ be a ring and $n = lm$ where $l,m$ are positive integers. We denote the $(\lambda,l)$-quasi-twisted shift operator by $\theta_{\lambda,l}$ and it is defined on $R^n$ by$$\theta_{\lambda,l}(x_0\vert x_1\vert\ldots\vert x_{l-1}) = (\sigma_{\lambda}(x_0)\vert\sigma_{\lambda}(x_1)\vert\ldots\vert\sigma_{\lambda}(x_{l-1})),$$where $\lambda\in R$ is a unit, $x_i\in R^m$ for $i = 0,1,\ldots,l-1$ and $\sigma_{\lambda}$ is the $\lambda$-constacyclic shift operator on $R^m$. An $R$-additive code $\mathcal{C}$ of length $n$ is said to be a $(\lambda,l)$-quasi-twisted code if $\mathcal{C}$ is invariant under the map $\theta_{\lambda,l}$.
\end{definition}

\begin{definition}
Let $R$ be a ring and $R_i = \frac{R[x]}{\langle x^{m_i}-\lambda_i\rangle},~i=1,2,\ldots,l$ where $m_1,m_2,\ldots,m_l$ are positive integers and $\lambda_1,\lambda_2,\ldots,\lambda_l$ are units in $R$. Then any $R[x]$-submodule of $R_1\times R_2\times\cdots\times R_l$ is called a generalized $(\lambda_1,\lambda_2,\ldots,\lambda_l)$-quasi-twisted code of block length $(m_1,m_2,\ldots,m_l)$.\vspace{0.01 cm}

We can observe that a generalized $(\lambda_1,\lambda_2,\ldots,\lambda_l)$-quasi-twisted code of block length $(m_1,m_2,\ldots,m_l)$, where $\lambda_1=\lambda_2=\cdots=\lambda_l=\lambda,~m_1=m_2=\cdots=m_l=m$, is a $(\lambda,l)$-quasi-twisted code of length $lm$.
\end{definition}

Now, by using the standard inner product (Euclidean inner product), we define the dual of a code, self-orthogonal code, self-dual code and dual containing code.\vspace{0.01 cm}

 \begin{definition}
 Let $R$ be a ring and $\mathcal{C}$ be a linear code of length $n$ over $R$. Then the dual code $\mathcal{C}^\perp$ of code $\mathcal{C}$ is defined as $$\mathcal{C}^\perp =\{v\in R^n \mid v\cdot c=0 \text{ for all } c \in \mathcal{C}\}.$$
 A code $\mathcal{C}$ is called self-orthogonal, dual-containing and self-dual code if it satisfies $\mathcal{C}\subseteq \mathcal{C}^\perp$, $\mathcal{C}^\perp \subseteq \mathcal{C}$ and $\mathcal{C}= \mathcal{C}^\perp$, respectively.
 \end{definition}

 \section{Constacyclic codes over $\mathcal{R}$, $\mathcal{S}$, $\mathcal{R}\mathcal{S}$ and $\mathbb{Z}_p\mathcal{R}\mathcal{S}$}

 In this section, first, we study additive constacyclic codes over $\mathcal{R}$ and $\mathcal{S}$, and then generalize the results over $\mathcal{R}\mathcal{S}$ and $\mathbb{Z}_p\mathcal{R}\mathcal{S}$. Here, our main objective is to find the generators of additive constacyclic codes. We define suitable inner products, and under those inner products, we also find the generators of the dual of additive constacyclic codes.

 \subsection{\textbf{Constacyclic codes over }$\mathbf{\mathcal{R}}$\textbf{ and }$\mathbf{\mathcal{S}}$}\label{Section3.1}

Here, we first study all the units of $\mathcal{R}$ and $\mathcal{S}$, and then we find the generators of constacyclic codes over $\mathcal{R}$ and $\mathcal{S}$ respectively.\vspace{0.2 cm}

Let us denote the group of units of $\mathcal{R}$ and $\mathcal{S}$ by $U(\mathcal{R})$ and $U(\mathcal{S}),$ respectively. Then $U(\mathcal{R}) = \{a+ub\mid a,b\in \mathbb{Z}_p ~\text{and}~ a\neq 0\}$ and $U(\mathcal{S}) = \{a+ub+u^2d\mid a,b,d\in \mathbb{Z}_p ~\text{and}~ a\neq 0\}$. We denote the set of all nonzero elements of $\mathbb{Z}_p$ by $\mathbb{Z}_p^{\ast}$.

Now, we define a few maps.\vspace{0.2 cm}

(a) Define $\eta_0 : \mathcal{R} \rightarrow \mathbb{Z}_p$ by $\eta_0(a+bu) = a$.\vspace{0.2 cm}

(b) Define $\eta_1 : \mathcal{S} \rightarrow \mathbb{Z}_p$ by $\eta_1(a+bu+du^2) = a$.\vspace{0.2 cm}

(c) Define $\eta_2 : \mathcal{S} \rightarrow \mathcal{R}$ by $\eta_2(a+bu+du^2) = a+bu$.\vspace{0.2 cm}

The following lemma shows the interrelation between the units of $\mathbb{Z}_p, \mathcal{R}$ and $\mathcal{S}$.

\begin{lemma}
Let $\mu_1 \in \mathcal{R}$ and $\mu_2\in \mathcal{S}$. Then we have the following.
\begin{enumerate}
\item $\mu_1 \in U(\mathcal{R})$ if and only if $\eta_0(\mu_1) \in \mathbb{Z}_p^{\ast}$;
\item $\mu_2 \in U(\mathcal{S})$ if and only if $\eta_1(\mu_2) \in \mathbb{Z}_p^{\ast}$;
\item $\mu_2 \in U(\mathcal{S})$ if and only if $\eta_2(\mu_2) \in U(\mathcal{R})$.
\end{enumerate}
\end{lemma}
Now, we recall one result from \cite{AbhaySir}.
\begin{theorem}\label{Theorem1}
Let $\mathfrak{C}$ be a cyclic code of length $n$ over $\frac{\mathbb{Z}_p[u]}{\langle u^k\rangle}$. If $n$ is relatively prime to $p$ then $\mathfrak{C} = \langle f_0(x)+uf_1(x)+\cdots+u^{k-1}f_{k-1}(x)\rangle$ where $f_0(x), f_1(x),\ldots,f_{k-1}(x)\in \mathbb{Z}_p[x]$ and $f_{k-1}(x)\mid f_{k-2}(x)\mid\cdots\mid f_0(x)\mid(x^n-1)$ in $\mathbb{Z}_p[x]$.
\end{theorem}
The following theorem gives the generator of a constacyclic code over $\mathcal{R}$.

\begin{theorem}\label{Theorem2}
Let $\mathcal{R} = \frac{\mathbb{Z}_p[u]}{\langle u^2\rangle}$ and $\mu_1\in U(\mathcal{R})$. Then for any positive integer $r$ satisfying $\gcd(p,r)=1$ and $r\equiv 1 ~(\text{mod}~ \text{ord}{(\mu_1)})$, every $\mu_1$-constacyclic code $\mathfrak{C}$ of length $r$ over $\mathcal{R}$ is given by $$\mathfrak{C} = \langle f_0(x)+uf_1(x)\rangle$$ where $f_1(x)\mid f_0(x)\mid (x^r-\mu_1)$ over $\mathcal{R}$.
\end{theorem}
\begin{proof}
We define a map $\rho_{\mathcal{R}} : \frac{\mathcal{R}[x]}{\langle x^r-1\rangle} \rightarrow \frac{\mathcal{R}[x]}{\langle x^r-\mu_1\rangle}$ by $\rho_{\mathcal{R}} (f(x)) = f(\mu_1^{-1}x)$. Let $a(x), b(x)$ and $h(x)$ be three polynomials in $\mathcal{R}[x]$ such that $a(x)-b(x) = (x^r-1)h(x)$. Now, $a(x)-b(x) = (x^r-1)h(x) \Leftrightarrow a(\mu_1^{-1}x)-b(\mu_1^{-1}x) = ((\mu_1^{-1}x)^r-1)h(\mu_1^{-1}x) \Leftrightarrow a(\mu_1^{-1}x)-b(\mu_1^{-1}x) = \mu_1^{-1}(x^r-\mu_1)h(\mu_1^{-1}x)$. This shows that $\rho_{\mathcal{R}}$ is an isomorphism of rings and if $I$ is an ideal of $\frac{\mathcal{R}[x]}{\langle x^r-1\rangle}$ then $\rho_{\mathcal{R}}(I)$ is an ideal of $\frac{\mathcal{R}[x]}{\langle x^r-\mu_1\rangle}$. Thus if $\mathfrak{C}$ is a $\mu_1$-constacyclic code over $\mathcal{R},$ then $\rho_\mathcal{R}^{-1}(\mathfrak{C})$ is a cyclic code over $\mathcal{R}$ and by Theorem $\ref{Theorem1}$, $\rho_\mathcal{R}^{-1}(\mathfrak{C}) = \langle a_0(x)+ua_1(x)\rangle$ where $a_1(x)\mid a_0(x)$ and $a_0(x)\mid (x^r-1)$ over $\mathcal{R}$. Hence, $\mathfrak{C} = \langle a_0(\mu_1^{-1}x)+ua_1(\mu_1^{-1}x)\rangle$. Consider $f_0(x) = a_0(\mu_1^{-1}x)$ and $f_1(x) = a_1(\mu_1^{-1}x)$. Then one can easily verify that $f_1(x)\mid f_0(x)\mid (x^r-\mu_1)$ over $\mathcal{R}$.
\end{proof}
Next, we find the generator of a constacyclic code over $\mathcal{S}$.

\begin{theorem}
Let $\mathcal{S} = \frac{\mathbb{Z}_p[u]}{\langle u^3\rangle}$ and $\mu_2\in U(\mathcal{S})$. Then for any positive integer $s$ satisfying $\gcd(p,s)=1$ and $s\equiv 1 ~(\text{mod}~ \text{ord}{(\mu_2)})$, every $\mu_2$-constacyclic code $\mathfrak{C}$ of length $s$ over $\mathcal{S}$ is given by $$\mathfrak{C} = \langle f_0(x)+uf_1(x)+u^2f_2(x)\rangle$$ where $f_2(x)\mid f_1(x)\mid f_0(x)\mid (x^s-\mu_2)$ over $\mathcal{S}$.
\end{theorem}
\begin{proof}
We define a map $\rho_{\mathcal{S}} : \frac{\mathcal{S}[x]}{\langle x^s-1\rangle} \rightarrow \frac{\mathcal{S}[x]}{\langle x^s-\mu_2\rangle}$ by $\rho_{\mathcal{S}} (f(x)) = f(\mu_2^{-1}x)$. Let $a(x), b(x)$ and $h(x)$ be three polynomials in $\mathcal{S}[x]$ such that $a(x)-b(x) = (x^s-1)h(x)$. Now, $a(x)-b(x) = (x^s-1)h(x) \Leftrightarrow a(\mu_2^{-1}x)-b(\mu_2^{-1}x) = ((\mu_2^{-1}x)^s-1)h(\mu_2^{-1}x) \Leftrightarrow a(\mu_2^{-1}x)-b(\mu_2^{-1}x) = \mu_2^{-1}(x^s-\mu_2)h(\mu_2^{-1}x)$. This shows that $\rho_\mathcal{S}$ is an isomorphism of rings and if $I$ is an ideal of $\frac{\mathcal{S}[x]}{\langle x^s-1\rangle}$ then $\rho_\mathcal{S}(I)$ is an ideal of $\frac{\mathcal{S}[x]}{\langle x^s-\mu_2\rangle}$. Thus if $\mathfrak{C}$ is a $\mu_2$-constacyclic code over $\mathcal{S}$ then $\rho_\mathcal{S}^{-1}(\mathfrak{C})$ is a cyclic code over $\mathcal{S}$ and by Theorem $\ref{Theorem1}$, $\rho_\mathcal{S}^{-1}(\mathfrak{C}) = \langle a_0(x)+ua_1(x)+u^2a_2(x)\rangle$ where $a_2(x)\mid a_1(x)\mid a_0(x)$ and $a_0(x)\mid (x^s-1)$ over $\mathcal{S}$. Hence $\mathfrak{C} = \langle a_0(\mu_2^{-1}x)+ua_1(\mu_2^{-1}x)+u^2a_2(\mu_2^{-1}x)\rangle$. Consider $f_0(x) = a_0(\mu_2^{-1}x), f_1(x) = a_1(\mu_2^{-1}x)$ and $f_2(x) = a_2(\mu_2^{-1}x)$. Then one can easily verify that $f_2(x)\mid f_1(x)\mid f_0(x)\mid (x^s-\mu_1)$ over $\mathcal{S}$.
\end{proof}

\subsection{\textbf{Constacyclic codes over }$\mathbf{\mathbb{Z}_p\mathcal{R}\mathcal{S}}$}

Let $\mu_0\in \mathbb{Z}_p^{\ast},~ \mu_1\in U(\mathcal{R}),~ \mu_2\in U(\mathcal{S})$ and $q\equiv 1 (\text{mod ord} (\mu_0)),~ r\equiv 1 (\text{mod ord} (\mu_1)),~ s\equiv 1 (\text{mod ord} (\mu_2)),~ \gcd{(p,q)} = 1,~ \gcd{(p,r)} = 1$ and $\gcd{(p,s)} = 1$. From here onwards, we will continue with these conditions. One can easily check that the three maps $\eta_0, \eta_1$ and $\eta_2$, defined in section \ref{Section3.1}, are ring epimorphisms and using these maps, we can define the module structures of $\frac{\mathbb{Z}_p[x]}{\langle x^q-\mu_0\rangle}\times \frac{\mathcal{R}[x]}{\langle x^r-\mu_1\rangle}$, $\frac{\mathcal{R}[x]}{\langle x^r-\mu_1\rangle}\times \frac{\mathcal{S}[x]}{\langle x^s-\mu_2\rangle}$ and $\frac{\mathbb{Z}_p[x]}{\langle x^q-\mu_0\rangle}\times \frac{\mathcal{R}[x]}{\langle x^r-\mu_1\rangle}\times \frac{\mathcal{S}[x]}{\langle x^s-\mu_2\rangle}$. Let us denote $\mathcal{M}_{0,1}^{q,r} = \frac{\mathbb{Z}_p[x]}{\langle x^q-\mu_0\rangle}\times \frac{\mathcal{R}[x]}{\langle x^r-\mu_1\rangle}$, $\mathcal{M}_{1,2}^{r,s} = \frac{\mathcal{R}[x]}{\langle x^r-\mu_1\rangle}\times \frac{\mathcal{S}[x]}{\langle x^s-\mu_2\rangle}$ and $\mathcal{M}_{0,1,2}^{q,r,s} = \frac{\mathbb{Z}_p[x]}{\langle x^q-\mu_0\rangle}\times \frac{\mathcal{R}[x]}{\langle x^r-\mu_1\rangle}\times \frac{\mathcal{S}[x]}{\langle x^s-\mu_2\rangle}$. Then $\mathcal{M}_{0,1}^{q,r}$ has an $\mathcal{R}[x]$-module structure with the scalar multiplication defined as $$g(x)\cdot (a(x),b(x)) = (\eta_0(g(x))a(x),g(x)b(x)).$$ Similarly, $\mathcal{M}_{1,2}^{r,s}$ has an $\mathcal{S}[x]$-module structure with the scalar multiplication defined as $$h(x)\cdot (b(x),d(x)) = (\eta_1(h(x))b(x),h(x)d(x))$$ and $\mathcal{M}_{0,1,2}^{q,r,s}$ has an $\mathcal{S}[x]$-module structure with the scalar multiplication defined as $$h(x)\cdot (a(x),b(x),d(x)) = (\eta_1(h(x))a(x),\eta_2(h(x))b(x),h(x)d(x)),$$ where $g(x)\in \mathcal{R}[x],~ h(x)\in \mathcal{S}[x],~ a(x)\in \frac{\mathbb{Z}_p[x]}{\langle x^q-\mu_0\rangle},~ b(x)\in \frac{\mathcal{R}[x]}{\langle x^r-\mu_1\rangle},~ d(x)\in \frac{\mathcal{S}[x]}{\langle x^s-\mu_2\rangle}$. Note that if $g(x) = \displaystyle\sum_{j=0}^l g_jx^j$ then $\eta_0(g(x)) = \displaystyle\sum_{j=0}^l \eta_0(g_j)x^j$ and if $h(x) = \displaystyle\sum_{j=0}^l h_jx^j,$ then $\eta_i(h(x)) = \displaystyle\sum_{j=0}^l \eta_i(h_j)x^j$ for $i=1,2$.

\begin{definition}
The operator $T_{\mu_0,\mu_1} : \mathbb{Z}_p^q \mathcal{R}^r\rightarrow \mathbb{Z}_p^q \mathcal{R}^r$ defined by $T_{\mu_0,\mu_1}(a_0,a_1,\ldots,a_{q-1}\vert b_0,b_1,\ldots,b_{r-1})$ $= (\mu_0a_{q-1},a_0,a_1,\ldots,a_{q-2}\vert \mu_1b_{r-1},b_0,b_1,\ldots,b_{r-2})$ is called the $(\mu_0,\mu_1)$-constacyclic shift operator. A $\mathbb{Z}_p\mathcal{R}$-additive code $\mathfrak{C}$ having block length $(q,r)$ is said to be $(\mu_0,\mu_1)$-constacyclic code if $T_{\mu_0,\mu_1}(\mathfrak{C})\subset \mathfrak{C}$. Similarly, one can define a $(\mu_1,\mu_2)$-constacyclic code of block length $(r,s)$ over  $\mathcal{R}\mathcal{S}$ with the $(\mu_1,\mu_2)$-constacyclic shift operator $T_{\mu_1,\mu_2}$ defined over $\mathcal{R}^r \mathcal{S}^s$ and a $(\mu_0,\mu_1,\mu_2)$-constacyclic code of block length $(q,r,s)$ over $\mathbb{Z}_p \mathcal{R}\mathcal{S}$ with the $(\mu_0,\mu_1,\mu_2)$-constacyclic shift operator $T_{\mu_0,\mu_1,\mu_2}$ defined over $\mathbb{Z}_p^q \mathcal{R}^r\mathcal{S}^s$, respectively. For $q=0$ and $q=0,r=0$, it is said to be $(\mu_1,\mu_2)$-constacyclic code over $\mathcal{R}\mathcal{S}$ and $\mu_2$-constacylic code over $\mathcal{S},$ respectively.
\end{definition}

Let $\zeta_0 : \mathbb{Z}_p^q \mathcal{R}^r\rightarrow \mathcal{M}_{0,1}^{q,r}$ be defined by $\zeta_0(c) = c(x)$ where $c=(a_0,a_1,\ldots,a_{q-1}\vert b_0,b_1,\ldots,b_{r-1})$ and $c(x)=(a(x),b(x))$ with $a(x) = \displaystyle\sum_{i=0}^{q-1}a_ix^i,~ b(x) = \displaystyle\sum_{i=0}^{r-1}b_ix^i$. It is easy to observe that $\zeta_0$ is an isomorphism. Now, we have the following lemma.

\begin{lemma}\label{Lemma4.1}
The isomorphism $\zeta_0$ maps the constacyclic shift $T_{\mu_0,\mu_1}$ of an element $c$ of $\mathbb{Z}_p\mathcal{R}$ to the multiplication of its image by $x$, i.e.,$$\zeta_0(T_{\mu_0,\mu_1}(c)) = x\cdot \zeta_0(c).$$
\end{lemma}
Using the above lemma, we have two equivalent statements.
\begin{proposition}
The following two statements are equivalent:
\begin{enumerate}
\item The $\mathbb{Z}_p\mathcal{R}$-additive code $\mathfrak{C}$ having block length $(q,r)$ is a $(\mu_0,\mu_1)$-constacyclic code.
\item $\zeta_0(\mathfrak{C})$ is a submodule of $\mathcal{M}_{0,1}^{q,r}$ over $\mathcal{R}[x]$.
\end{enumerate}
\end{proposition}
\begin{proof}
Let $\mathfrak{C}$ be an $\mathcal{R}$-submodule of $\mathbb{Z}_p^q \mathcal{R}^r$. Then by Lemma \ref{Lemma4.1}, $T_{\mu_0,\mu_1}(\mathfrak{C})\subset \mathfrak{C}\Leftrightarrow x\cdot\zeta_0(\mathfrak{C})\subset\zeta_0(\mathfrak{C})\Leftrightarrow f(x)\cdot\zeta_0(\mathfrak{C})\subset\zeta_0(\mathfrak{C})$ for all $f(x)\in\mathcal{R}[x]$.
\end{proof}
Note that in a similar way as above, a $(\mu_1,\mu_2)$-constacyclic code over $\mathcal{R}\mathcal{S}$ and a $(\mu_0,\mu_1,\mu_2)$-constacyclic code over $\mathbb{Z}_p\mathcal{R}\mathcal{S}$ can be considered as an $\mathcal{S}[x]$-submodule of $\mathcal{M}_{1,2}^{r,s}$ and an $\mathcal{S}[x]$-submodule of $\mathcal{M}_{0,1,2}^{q,r,s}$, respectively.\vspace{0.1 cm}

In the following theorem, we find the generator of a $(\mu_1,\mu_2)$-constacyclic code over $\mathcal{R}\mathcal{S}$.
\begin{theorem}\label{Theorem4}
Any $(\mu_1,\mu_2)$-constacyclic code $\mathfrak{C}$ of block length $(r,s)$ over $\mathcal{R}\mathcal{S}$ can be given as $$\mathfrak{C} = \langle (g_0(x)+ug_1(x),0),(l(x),h_0(x)+uh_1(x)+u^2h_2(x))\rangle$$ where $g_0(x),g_1(x),h_0(x),h_1(x),h_2(x)$ are polynomials over $\mathbb{Z}_p$ satisfying $g_1(x)\mid g_0(x)\mid (x^r-\mu_1)$ over $\mathcal{R}$, $h_2(x)\mid h_1(x)\mid h_0(x)\mid (x^s-\mu_2)$ over $\mathcal{S}$ and $l(x)\in\mathcal{R}[x]$ is such that $(l(x),h_0(x)+uh_1(x)+u^2h_2(x))\in\mathfrak{C}$.
\end{theorem}
\begin{proof}
Consider the projection map $p_2: \mathcal{M}_{1,2}^{r,s}\rightarrow\frac{\mathcal{S}[x]}{\langle x^s-\mu_2\rangle}$ defined by $p_2(a(x),b(x)) = b(x)$. It is verified that $p_2$ is an $\mathcal{S}[x]$-linear map. Denote $\hat{p}_2 = p_2\vert_{\mathfrak{C}}$. Since $\mathfrak{C}$ is an $\mathcal{S}[x]$-submodule of $\mathcal{M}_{1,2}^{r,s}$, $\hat{p}_2(\mathfrak{C})$ is also an $\mathcal{S}[x]$-submodule of $\frac{\mathcal{S}[x]}{\langle x^s-\mu_2\rangle}$. Thus $\hat{p}_2(\mathfrak{C})$ is a $\mu_2$-constacyclic code of length $s$ over $\mathcal{S}$ and hence$$\hat{p}_2(\mathfrak{C}) = \langle h_0(x)+uh_1(x)+u^2h_2(x)\rangle$$where $h_2(x)\mid h_1(x)\mid h_0(x)\mid (x^s-\mu_2)$ over $\mathcal{S}$. Also, $\text{ker}(\hat{p}_2) = \{(a(x),0)\in\mathfrak{C}\mid a(x)\in\frac{\mathcal{R}[x]}{\langle x^r-\mu_1\rangle}\}$ is an $\mathcal{S}[x]$-submodule of $\mathfrak{C}$. Let $I = \{a(x)\in\frac{\mathcal{R}[x]}{\langle x^r-\mu_1\rangle}\mid (a(x),0)\in \text{ker}(\hat{p}_2)\}$. Then $I$ is an $\mathcal{R}[x]$-submodule of $\frac{\mathcal{R}[x]}{\langle x^r-\mu_1\rangle}$ and so $I$ is a $\mu_1$-constacyclic code of length $r$ over $\mathcal{R}$. From Theorem \ref{Theorem2}, we have $I = \langle g_0(x)+ug_1(x)\rangle$ where $g_1(x)\mid g_0(x)\mid (x^r-\mu_1)$ over $\mathcal{R}$. Therefore, $$\text{ker}(\hat{p}_2) = \langle (g_0(x)+ug_1(x),0)\rangle.$$Let $l(x)\in\mathcal{R}[x]$ be such that $(l(x),h_0(x)+uh_1(x)+u^2h_2(x))\in\mathfrak{C}$. Denote $c_1(x) = (l(x),h_0(x)+uh_1(x)+u^2h_2(x))$ and $c_2(x) = (g_0(x)+ug_1(x),0)$. Take $c(x) = (a_1(x),a_2(x))\in\mathfrak{C}$. Then $a_2(x)\in\hat{p}_2(\mathfrak{C})$ and thus there exists a polynomial $q_1(x)\in\mathcal{S}[x]$ such that $a_2(x) = q_1(x)(h_0(x)+uh_1(x)+u^2h_2(x))$. Now, $c(x)-q_1(x)c_1(x) = (a_1(x)-\eta_2(q_1(x))l(x),0)\in\text{ker}(\hat{p}_2)$ and hence $a_1(x)-\eta_2(q_1(x))l(x) = q_2(x)(g_0(x)+ug_1(x))$ for some $q_2(x)\in\mathcal{R}[x]$. Thus, $c(x) = q_1(x)c_1(x)+q_2(x)c_2(x)$.
\end{proof}

Now, we define an inner product on $\mathcal{R}^r\mathcal{S}^s$.

\begin{definition}\label{Definition5}
 Let $v=(y_0,y_1, \ldots, y_{r-1}|z_0,z_1, \ldots, z_{s-1}), w=(y_0',y_1', \ldots, y_{r-1}'|z_0',z_1', \ldots, z_{s-1}')\in \mathcal{R}^r\mathcal{S}^s$. Define the inner product of $v$ and $w$ by $$\langle v,w\rangle=u\sum_{i=0}^{r-1}y_iy_i'+\sum_{i=0}^{s-1}z_iz_i'.$$
 From hereon, for an $\mathcal{R}\mathcal{S}$-additive code $\frak{C}$, its dual will be defined with respect to this inner product, and the dual will be denoted by $\frak{C}^\perp$. Note that the block length of $\frak{C}^\perp$ is the same as the block length of $\frak{C}$.
 \end{definition}

 \begin{theorem}\label{Theorem5}
 Let $\mu_1\in U(\mathcal{R})$, $\mu_2 \in U(\mathcal{S})$. Then the code $\frak{C}$ is an $\mathcal{R}\mathcal{S}$-additive $(\mu_1,\mu_2)$-constacyclic code iff $\frak{C}^\perp$ is an $\mathcal{R}\mathcal{S}$-additive $(\mu_1^{-1},\mu_2^{-1})$-constacyclic code.
 \end{theorem}
 \begin{proof}
Let $\frak{C}$ be a $(\mu_1,\mu_2)$-constacyclic code of block length $(r,s)$. Let $c=(y_0,y_1, \ldots, y_{r-1}\vert z_0,z_1, \ldots, z_{s-1})$, $c^{\prime}=(y_0^{\prime},y_1^{\prime}, \ldots, y_{r-1}^{\prime}\vert z_0^{\prime},z_1^{\prime}, \ldots, z_{s-1}^{\prime})\in \frak{C}^\perp$. Now take $l_1$ and $l_2$ as the order of $\mu_1$ and $\mu_2$ respectively and fix $m=l_1l_2rs$. Then $T_{\mu_1,\mu_2}^m(c)=c$ and hence$$T_{\mu_1,\mu_2}^{m-1}(c)=(y_1, \ldots, y_{r-2},y_{r-1},\mu_1^{-1}y_0|z_1, \ldots, z_{s-2},z_{s-1},\mu_2^{-1}z_0).$$
 Since $c \in \frak{C}$, then $\frak{C}$ being $(\mu_1,\mu_2)$-constacyclic code we get $T_{\mu_1,\mu_2}^{m-1}(c) \in \frak{C}$.\\
 Now for given $c \text{ and } c^{\prime}$, we get
 \begin{align*}
     \langle T_{\mu_1^{-1},\mu_2^{-1}}^{m-1}(c^{\prime}),c \rangle&=u\{\mu_1^{-1}y_{r-1}^{\prime}y_0+y_0^{\prime}y_1+\cdots + y_{r-2}^{\prime}y_{r-1}\}+\{\mu_2^{-1}z_{s-1}^{\prime}z_0+z_0^{\prime}z_1+\cdots + z_{s-2}^{\prime}z_{s-1}\} \\
     &=u\{y_0^{\prime}y_1+y_1^{\prime}y_2+\cdots + \mu_1^{-1}y_{r-1}^{\prime}y_0\}+\{z_0^{\prime}z_1+z_1^{\prime}z_2+\cdots +\mu_2^{-1}z_{s-1}^{\prime}z_0\} \\
     &= \langle c^{\prime},T_{\mu_1,\mu_2}^{m-1}(c) \rangle\\
     &=0
 \end{align*}
 Since $T_{\mu_1^{-1},\mu_2^{-1}}(c^{\prime})\in \frak{C}^\perp$, code $\frak{C}^\perp$ is a $(\mu_1^{-1},\mu_2^{-1})$-constacylic code.\\
By changing the role of $\frak{C}$ and $\frak{C}^\perp$, we get the converse of the theorem.
 \end{proof}

\begin{corollary}
If $\mathfrak{C}$ is any $(\mu_1,\mu_2)$-constacyclic code  of block length $(r,s)$ over $\mathcal{R}\mathcal{S},$ then $\mathfrak{C}^{\perp}$ is a $(\mu_1^{-1},\mu_2^{-1})$-constacyclic code of block length $(r,s)$ over $\mathcal{R}\mathcal{S}$ and$$\mathfrak{C}^{\perp} = \langle (g_0^{\prime}(x)+ug_1^{\prime}(x),0),(l^{\prime}(x),h_0^{\prime}(x)+uh_1^{\prime}(x)+u^2h_2^{\prime}(x))\rangle$$ where $g_0^{\prime}(x),g_1^{\prime}(x),h_0^{\prime}(x),h_1^{\prime}(x),h_2^{\prime}(x)$ are polynomials over $\mathbb{Z}_p$, satisfying $g_1^{\prime}(x)\mid g_0^{\prime}(x)\mid (x^r-\mu_1)$ over $\mathcal{R}$, $h_2^{\prime}(x)\mid h_1^{\prime}(x)\mid h_0^{\prime}(x)\mid (x^s-\mu_2)$ over $\mathcal{S}$ and $l^{\prime}(x)\in\mathcal{R}[x]$ is such that $(l^{\prime}(x),h_0^{\prime}(x)+uh_1^{\prime}(x)+u^2h_2^{\prime}(x))\in\mathfrak{C}$.
\end{corollary}
\begin{proof}
It follows directly from Theorem \ref{Theorem4} and Theorem \ref{Theorem5}.
\end{proof}
In the next result, we find the generator of a $(\mu_0,\mu_1,\mu_2)$-constacyclic code over $\mathbb{Z}_p\mathcal{R}\mathcal{S}$.
\begin{theorem}\label{Theorem6}
Any $(\mu_0,\mu_1,\mu_2)$-constacyclic code $\mathfrak{C}$ of block length $(q,r,s)$ over $\mathbb{Z}_p\mathcal{R}\mathcal{S}$ can be given as$$\mathfrak{C} = \langle(f_0(x),0,0),(l_1(x),g_0(x)+ug_1(x),0),(l_2(x),l_3(x),h_0(x)+uh_1(x)+u^2h_2(x))\rangle$$where $f_0(x),g_0(x),g_1(x),h_0(x),h_1(x),h_2(x)$ are polynomials over $\mathbb{Z}_p$, satisfying $f_0(x)\mid (x^q-\mu_0)$ over $\mathbb{Z}_p$, $g_1(x)\mid g_0(x)\mid (x^r-\mu_1)$ over $\mathcal{R}$, $h_2(x)\mid h_1(x)\mid h_0(x)\mid (x^s-\mu_2)$ over $\mathcal{S}$ and the polynomials $l_1(x),l_2(x)\in\mathbb{Z}_p[x]$, $l_3(x)\in\mathcal{R}[x]$ are such that $(l_1(x),g_0(x)+ug_1(x),0),(l_2(x),l_3(x),h_0(x)+uh_1(x)+u^2h_2(x))\in\mathfrak{C}$.
\end{theorem}
\begin{proof}
Consider the projection map $p_3 : \mathcal{M}_{0,1,2}^{q,r,s}\rightarrow\frac{\mathcal{S}[x]}{\langle x^s-\mu_2\rangle}$ defined as $p_3(a(x),b(x),d(x)) = d(x)$. It can be verified that $p_3$ is an $\mathcal{S}[x]$-linear map. Denote $\hat{p}_3 = p_3\vert_{\mathfrak{C}}$. Since $\mathfrak{C}$ is an $\mathcal{S}[x]$-submodule of $\mathcal{M}_{0,1,2}^{q,r,s}$, $\hat{p}_3(\mathfrak{C})$ is also an $\mathcal{S}[x]$-submodule of $\frac{\mathcal{S}[x]}{\langle x^s-\mu_2\rangle}$. Therefore, $\hat{p}_3(\mathfrak{C})$ is a $\mu_2$-constacyclic code of length $s$ over $\mathcal{S}$ and hence$$\hat{p}_3(\mathfrak{C}) = \langle h_0(x)+uh_1(x)+u^2h_2(x)\rangle$$ where $h_2(x)\mid h_1(x)\mid h_0(x)\mid (x^s-\mu_2)$ over $\mathcal{S}$. Also, $\text{ker}(\hat{p}_3) = \{(a(x),b(x),0)\in\mathfrak{C}\mid a(x)\in\frac{\mathbb{Z}_p[x]}{\langle x^q-\mu_0\rangle},b(x)\in\frac{\mathcal{R}[x]}{\langle x^r-\mu_1\rangle}\}$ is an $\mathcal{S}[x]$-submodule of $\mathfrak{C}$. Let $I = \{(a(x),b(x))\in\mathcal{M}_{0,1}^{q,r}\mid (a(x),b(x),0)\in \text{ker}(\hat{p}_3)\}$. Then $I$ is an $\mathcal{R}[x]$-submodule of $\mathcal{M}_{0,1}^{q,r}$ and so $I$ is a $(\mu_0,\mu_1)$-constacyclic code of block length $(q,r)$ over $\mathbb{Z}_p\mathcal{R}$. Using similar arguments as in Theorem \ref{Theorem4}, we have $I = \langle(f_0(x),0),(l_1(x),g_0(x)+ug_1(x))\rangle$ where $f_0(x),l_1(x)$ are polynomials over $\mathbb{Z}_p$ with $f_0(x)\mid (x^q-\mu_0)$ over $\mathbb{Z}_p$ and $g_1(x)\mid g_0(x)\mid (x^r-\mu_1)$ over $\mathcal{R}$. Thus$$\text{ker}(\hat{p}_3) = \langle(f_0(x),0,0),(l_1(x),g_0(x)+ug_1(x),0)\rangle.$$Let $l_2(x)\in\mathbb{Z}_p[x]$, $l_3(x)\in\mathcal{R}[x]$ are such that $(l_2(x),l_3(x),h_0(x)+uh_1(x)+u^2h_2(x))\in\mathfrak{C}$. Denote $c_1(x) = (l_2(x),l_3(x),h_0(x)+uh_1(x)+u^2h_2(x)), c_2(x) = (f_0(x),0,0)$ and $c_3(x) = (l_1(x),g_0(x)+ug_1(x),0)$. Note that $c_2(x) = (f_0(x),0,0),c_3(x) = (l_1(x),g_0(x)+ug_1(x),0)\in\text{ker}(\hat{p}_3)\subset\mathfrak{C}$. Take $c(x) = (a_1(x),a_2(x),a_3(x))\in\mathfrak{C}$. Then $a_3(x)\in\hat{p}_3(\mathfrak{C})$ and thus there exists a polynomial $q_1(x)\in\mathcal{S}[x]$ such that $a_3(x) = q_1(x)(h_0(x)+uh_1(x)+u^2h_2(x))$. Now $c(x)-q_1(x)c_1(x) = (a_1(x)-\eta_1(q_1(x))l_2(x),a_2(x)-\eta_2(q_1(x))l_3(x),0)\in\text{ker}(\hat{p}_3)$. Hence, $a_1(x)-\eta_1(q_1(x))l_2(x) = q_2(x)(f_0(x)+q_3(x)l(x))$ for some $q_2(x)\in\mathbb{Z}_p[x]$ and $a_2(x)-\eta_2(q_1(x))l_3(x) = q_3(x)(g_0(x)+ug_1(x))$ for some $q_3(x)\in\mathcal{R}[x]$. Thus, $c(x) = q_1(x)c_1(x)+q_2(x)c_2(x)+q_3(x)c_3(x)$.
\end{proof}
Similarly, as in Definition \ref{Definition5}, we can define an inner product on $\mathbb{Z}_p^q\mathcal{R}^r\mathcal{S}^s$ as follows.
\begin{definition}\label{Definition6}
 Let $v=(x_0,x_1, \ldots, x_{q-1}\vert y_0,y_1, \ldots, y_{r-1}\vert z_0,z_1, \ldots, z_{s-1})$ and $w=(x_0^{\prime},x_1^{\prime}, \ldots, x_{q-1}^{\prime}\vert y_0^{\prime},\\y_1^{\prime}, \ldots, y_{r-1}^{\prime}\vert z_0^{\prime},z_1^{\prime}, \ldots, z_{s-1}^{\prime})$ are two members of $\mathbb{Z}_p^q\mathcal{R}^r\mathcal{S}^s$. Define the inner product of $v$ and $w$ by $$<v,w>~=u^2\sum_{i=0}^{q-1}x_ix_i'+u\sum_{i=0}^{r-1}y_iy_i'+\sum_{i=0}^{s-1}z_iz_i'.$$ From here onwards, the dual of a $\mathbb{Z}_p\mathcal{R}\mathcal{S}$-additive code will be defined with respect to this inner product.
 \end{definition}

\begin{theorem}\label{Theorem7}
The code $\frak{C}$ is a $\mathbb{Z}_p\mathcal{R}\mathcal{S}$-additive $(\mu_0,\mu_1,\mu_2)$-constacyclic code iff $\frak{C}^\perp$ is a $\mathbb{Z}_p\mathcal{R}\mathcal{S}$-additive $(\mu_0^{-1},\mu_1^{-1},\mu_2^{-1})$-constacyclic code.
\end{theorem}
\begin{proof}
By using Definition \ref{Definition6}, the proof follows similarly as in Theorem \ref{Theorem5}.
 \end{proof}

The following corollary is an obvious consequence of Theorem \ref{Theorem6} and Theorem \ref{Theorem7}.
\begin{corollary}
If $\mathfrak{C}$ is any $(\mu_0,\mu_1,\mu_2)$-constacyclic code of block length $(q,r,s)$ over $\mathbb{Z}_p\mathcal{R}\mathcal{S}$, then $\mathfrak{C}^{\perp}$ is a $(\mu_0^{-1},\mu_1^{-1},\mu_2^{-1})$-constacyclic code of block length $(q,r,s)$ over $\mathbb{Z}_p\mathcal{R}\mathcal{S}$ and$$\mathfrak{C}^{\perp} = \langle(f_0^{\prime}(x),0,0),(l_1^{\prime}(x),g_0^{\prime}(x)+ug_1^{\prime}(x),0),(l_2^{\prime}(x),l_3^{\prime}(x),h_0^{\prime}(x)+uh_1^{\prime}(x)+u^2h_2^{\prime}(x))\rangle$$where $f_0^{\prime}(x),g_0^{\prime}(x),g_1^{\prime}(x),h_0^{\prime}(x),h_1^{\prime}(x),h_2^{\prime}(x)$ are polynomials over $\mathbb{Z}_p$, satisfying $f_0^{\prime}(x)\mid (x^q-\mu_0)$ over $\mathbb{Z}_p$, $g_1^{\prime}(x)\mid g_0^{\prime}(x)\mid (x^r-\mu_1)$ over $\mathcal{R}$, $h_2^{\prime}(x)\mid h_1^{\prime}(x)\mid h_0^{\prime}(x)\mid (x^s-\mu_2)$ over $\mathcal{S}$ and the polynomials $l_1^{\prime}(x),l_2^{\prime}(x)\in\mathbb{Z}_p[x]$, $l_3^{\prime}(x)\in\mathcal{R}[x]$ are such that $(l_1^{\prime}(x),g_0^{\prime}(x)+ug_1^{\prime}(x),0),(l_2^{\prime}(x),l_3^{\prime}(x),h_0^{\prime}(x)+uh_1^{\prime}(x)+u^2h_2^{\prime}(x))\in\mathfrak{C}$.
\end{corollary}
\begin{definition}
Let $\mathfrak{C}$ be a $\mathbb{Z}_p\mathcal{R}\mathcal{S}$-additive code of block length $(q,r,s)$ and suppose $\mathfrak{C}_q$ be the code obtained by removing all the coordinates from $\mathcal{R}$ and $\mathcal{S}$, $\mathfrak{C}_r$ be the code obtained by removing all the coordinates from $\mathbb{Z}_p$ and $\mathcal{S}$, $\mathfrak{C}_s$ be the code obtained by removing all the coordinates from $\mathbb{Z}_p$ and $\mathcal{R}$. Then $\mathfrak{C}$ is called separable if $\mathfrak{C} = \mathfrak{C}_q\times\mathfrak{C}_r\times\mathfrak{C}_s$. If $\mathfrak{C}$ is separable then $\mathfrak{C}^{\perp} = \mathfrak{C}_q^{\perp}\times\mathfrak{C}_r^{\perp}\times\mathfrak{C}_s^{\perp}$, i.e., $\mathfrak{C}^{\perp}$ is also separable.
\end{definition}
It is easy to observe that $\mathfrak{C}_q$ is an additive code over $\mathbb{Z}_p$ of length $q$, $\mathfrak{C}_r$ is an additive code over $\mathcal{R}$ of length $r$ and $\mathfrak{C}_s$ is an additive code over $\mathcal{S}$ of length $s$. Now, we have the following two results.
\begin{proposition}\label{Proposition2}
Let $\mathfrak{C}$ be a $\mathbb{Z}_p\mathcal{R}\mathcal{S}$-additive code of block length $(q,r,s)$ and suppose $\mathfrak{C}$ is separable. Then $\mathfrak{C}$ is a $(\mu_0,\mu_1,\mu_2)$-constacyclic code if and only if $\mathfrak{C}_q$ is a $\mu_0$-constacyclic code, $\mathfrak{C}_r$ is a $\mu_1$-constacyclic code and $\mathfrak{C}_s$ is a $\mu_2$-constacyclic code.
\end{proposition}
\begin{proof}
Since $\mathfrak{C}$ is separable, we have $\mathfrak{C} = \mathfrak{C}_q\times\mathfrak{C}_r\times\mathfrak{C}_s$. First, suppose that $\mathfrak{C}$ is a $(\mu_0,\mu_1,\mu_2)$-constacyclic code. Take $(a_0,a_1,\ldots,a_{q-1})\in\mathfrak{C}_q, (b_0,b_1,\ldots,b_{r-1})\in\mathfrak{C}_r, (d_0,d_1,\ldots,d_{s-1})\in\mathfrak{C}_s$ such that $(a_0,a_1,\ldots,a_{q-1}\vert b_0,b_1,\ldots,b_{r-1}\vert d_0,d_1,\ldots,d_{s-1})\in\mathfrak{C}$. Then $(\mu_0a_{q-1},a_0,\ldots,a_{q-2}\vert \mu_1b_{r-1},\\b_0,\ldots,b_{r-2}\vert \mu_2d_{s-1},d_0,\ldots,d_{s-2})\in\mathfrak{C}$ and hence $(\mu_0a_{q-1},a_0,\ldots,a_{q-2})\in\mathfrak{C}_q, (\mu_1b_{r-1},b_0,\ldots,b_{r-2})\in\mathfrak{C}_r$ and $(\mu_2d_{s-1},d_0,\cdots,d_{s-2})\in\mathfrak{C}_s$. Thus $\mathfrak{C}_q$ is a $\mu_0$-constacyclic code, $\mathfrak{C}_r$ is a $\mu_1$-constacyclic code and $\mathfrak{C}_s$ is a $\mu_2$-constacyclic code.\vspace{0.1 cm}

The converse can be proved by reversing the above arguments.
\end{proof}
\begin{proposition}
Let $\mathfrak{C}$ be a $\mathbb{Z}_p\mathcal{R}\mathcal{S}$-additive code of block length $(q,r,s)$ and suppose $\mathfrak{C}$ is separable. Then $\mathfrak{C}$ is a $(\mu_0,\mu_1,\mu_2)$-constacyclic code if and only if $\mathfrak{C}^{\perp}$ is separable and it is a $(\mu_0^{-1}, \mu_1^{-1}, \mu_2^{-1})$-constacyclic code over $\mathbb{Z}_p\mathcal{R}\mathcal{S}$.
\end{proposition}
\begin{proof}
The proof follows from Theorem \ref{Theorem7}.
\end{proof}

\section{Gray maps on $\mathcal{R}$, $\mathcal{S}$ and $\mathbb{Z}_p\mathcal{R}\mathcal{S}$}
In this section, we define a few Gray maps and then study their images.\vspace{0.2 cm}

Define the Gray map$$\phi_1 : \mathcal{R}\rightarrow \mathbb{Z}_p^2$$ by
\begin{align*}
    \phi_1(1) &= (1,0)\\
    \phi_1(u) &= (1,\kappa),
\end{align*}
where $\kappa\in\mathbb{Z}_p$ is such that $\kappa^2\equiv -1~(\text{mod }p)$, and thus $\phi_1(a+ub) = (a+b,\kappa b)$ for all $a,b\in\mathbb{Z}_p$. Such an element $\kappa$ will always exist as we consider only those rings $\mathbb{Z}_p$ in which $p-1$ is a quadratic residue. One can easily verify that the map $\phi_1$ is $\mathbb{Z}_p$-linear and bijective.\vspace{0.05 cm}

The Lee weight of an element $x\in\mathcal{R}$ is defined as $wt_L(x) = wt_H(\phi_1(x))$, where $wt_H(y)$ denotes the Hamming weight of $y$ and the Lee distance between two elements $x,y\in\mathcal{R}$ is defined as $d_L(x,y) = wt_L(x-y)$. The Gray map $\phi_{1}$ can be extended to $\overline{\phi}_1 : \mathcal{R}^r\rightarrow\mathbb{Z}_p^{2r}$ by$$\overline{\phi}_1(a_0+ub_0,a_1+ub_1,\ldots,a_{r-1}+ub_{r-1}) = (a_0+b_0,a_1+b_1,\ldots,a_{r-1}+b_{r-1}\vert \kappa b_0,\kappa b_1,\ldots,\kappa b_{r-1}),$$where $\kappa\in\mathbb{Z}_p$ is such that $\kappa^2\equiv -1(\text{mod }p)$.\vspace{0.05 cm}

The Lee weight of an element $x = (x_0,x_1,\ldots,x_{r-1})\in\mathcal{R}^r$ is defined as $wt_L(x) = \displaystyle\sum_{i=0}^{r-1}wt_L(x_i)$ and the Lee distance between two elements $x,y\in\mathcal{R}^r$ is defined as $d_L(x,y) = wt_L(x-y)$.\vspace{0.05 cm}

From the above definition of $\overline{\phi}_1$, it can be observed that $\overline{\phi}_1$ is a distance preserving map from $\mathcal{R}^r$ (Lee weight) to $\mathbb{Z}_p^{2r}$ (Hamming weight). Further, if $\mathfrak{C}$ is an $\mathcal{R}$-additive code with parameters $[r,k,d]$ then $\overline{\phi}_1(\mathfrak{C})$ is a $[2r,k,d]$-code over $\mathbb{Z}_p$.
\begin{theorem}\label{Theorem8}
For any two elements $x,x^{\prime}\in\mathcal{R}^r$,$$\overline{\phi}_1(x)\cdot\overline{\phi}_1(x^{\prime}) = \pi_1\circ \phi_1(x\cdot x^{\prime}),$$where $\pi_1 : \mathbb{Z}_p^2\rightarrow\mathbb{Z}_p$ is the projection map defined as $\pi_1(a,b) = a$ for $a,b\in\mathbb{Z}_p$.
\end{theorem}
\begin{proof}
Let $x = (x_0,x_1,\ldots,x_{r-1})$ and $x^{\prime} = (x_0^{\prime},x_1^{\prime},\ldots,x_{r-1}^{\prime})$, where $x_i = a_i+ub_i,~ x_i^{\prime} = a_i^{\prime}+ub_i^{\prime}$ and $a_i,b_i,a_i^{\prime},b_i^{\prime}\in\mathbb{Z}_p$ for $i=0,1,\ldots,r-1$. Now,
\begin{align*}
    x\cdot x^{\prime}=&\sum_{i=0}^{r-1}x_ix_i^{\prime}\\
    =&\sum_{i=0}^{r-1}\{a_ia_i^{\prime}+u(a_ib_i^{\prime}+a_i^{\prime}b_i)\}\\
    \Rightarrow~ \phi_1(x\cdot x^{\prime}) =&~ (a_0a_0^{\prime}+a_0b_0^{\prime}+a_0^{\prime}b_0+\cdots+a_{r-1}a_{r-1}^{\prime}+a_{r-1}b_{r-1}^{\prime}+a_{r-1}^{\prime}b_{r-1}~,~ \kappa(a_0b_0^{\prime}+a_0^{\prime}b_0)+\\
    &~~\cdots+\kappa(a_{r-1}b_{r-1}^{\prime}+a_{r-1}^{\prime}b_{r-1}))
\end{align*}
Also,
\begin{align*}
    \overline{\phi}_1(x)\cdot\overline{\phi}_1(x^{\prime}) =~& (a_0+b_0,a_1+b_1,\ldots,a_{r-1}+b_{r-1}\vert \kappa b_0,\kappa b_1,\ldots,\kappa b_{r-1})\\
    \cdot&~(a_0^{\prime}+b_0^{\prime},a_1^{\prime}+b_1^{\prime},\ldots,a_{r-1}^{\prime}+b_{r-1}^{\prime}\vert \kappa b_0^{\prime},\kappa b_1^{\prime},\ldots,\kappa b_{r-1}^{\prime})\\
    =~& (a_0a_0^{\prime}+a_0b_0^{\prime}+a_0^{\prime}b_0+\cdots+a_{r-1}a_{r-1}^{\prime}+a_{r-1}b_{r-1}^{\prime}+a_{r-1}^{\prime}b_{r-1})
\end{align*}
Thus, $$\overline{\phi}_1(x)\cdot\overline{\phi}_1(x^{\prime}) = \pi_1\circ \phi_1(x\cdot x^{\prime}).$$
\end{proof}
From the above theorem, we get that the images of two orthogonal elements in $\mathcal{R}^r$ under the map $\overline{\phi}_1$ are also orthogonal. Now we have the following result.
\begin{corollary}\label{Corollary3}
Let $\mathfrak{C}$ be an $\mathcal{R}$-additive code of length $r$. Then$$\overline{\phi}_1(\mathfrak{C}^{\perp}) = \overline{\phi}_1(\mathfrak{C})^{\perp}.$$
\end{corollary}
\begin{proof}
From Theorem \ref{Theorem8}, we have $\overline{\phi}_1(\mathfrak{C}^{\perp})\subseteq\overline{\phi}_1(\mathfrak{C})^{\perp}$. Now, if we can show that $|\overline{\phi}_1(\mathfrak{C})^{\perp}| = |\overline{\phi}_1(\mathfrak{C}^{\perp})|$ then we are done. We know that $\overline{\phi}_1$ is an isomorphism and hence $|\overline{\phi}_1(\mathfrak{C})| = |\mathfrak{C}|$, $|\overline{\phi}_1(\mathfrak{C}^{\perp})| = |\mathfrak{C}^{\perp}|$. Also, $|\mathbb{Z}_p^{2r}| = p^{2r}$. Thus, $|\overline{\phi}_1(\mathfrak{C})^{\perp}| = \frac{p^{2r}}{|\overline{\phi}_1(\mathfrak{C})|} = \frac{p^{2r}}{|\mathfrak{C}|} = |\mathfrak{C}^{\perp}| = |\overline{\phi}_1(\mathfrak{C}^{\perp})|$.
\end{proof}

We define another Gray map$$\phi_2 : \mathcal{S}\rightarrow \mathbb{Z}_p^3$$ by
\begin{align*}
    \phi_2(1) &= (1,0,0)\\
    \phi_2(u) &= (1,\kappa,1)\\
    \phi_2(u^2) &= (1,\kappa,0),
\end{align*}
where $\kappa\in\mathbb{Z}_p$ is such that $\kappa^2\equiv -1~(\text{mod }p)$, and thus $\phi_2(a+ub+u^2d) = (a+b+d,\kappa(b+d),b)$ for all $a,b,d\in\mathbb{Z}_p$. One can easily verify that the map $\phi_2$ is $\mathbb{Z}_p$-linear and bijective.\vspace{0.05 cm}

The Lee weight of an element $x\in\mathcal{S}$ is defined as $wt_L(x) = wt_H(\phi_2(x))$, where $wt_H(y)$ denotes the Hamming weight of $y$ and the Lee distance between two elements $x,y\in\mathcal{S}$ is defined as $d_L(x,y) = wt_L(x-y)$. The Gray map $\phi_2$ can be extended to $\overline{\phi}_2 : \mathcal{S}^s\rightarrow\mathbb{Z}_p^{3s}$ by
\begin{align*}
    &\overline{\phi}_2(a_0+ub_0+u^2d_0,a_1+ub_1+u^2d_1,\ldots,a_{s-1}+ub_{s-1}+u^2d_{s-1})\\
    =~& (a_0+b_0+d_0,a_1+b_1+d_1,\ldots,a_{s-1}+b_{s-1}+d_{s-1}\vert \kappa(b_0+d_0),\kappa(b_1+d_1),\\
    &~~\ldots,\kappa(b_{s-1}+d_{s-1})\vert b_0,b_1,\ldots,b_{s-1})
\end{align*}
where $\kappa\in\mathbb{Z}_p$ is such that $\kappa^2\equiv -1(\text{mod }p)$.\vspace{0.05 cm}

The Lee weight of an element $x = (x_0,x_1,\ldots,x_{s-1})\in\mathcal{S}^s$ is defined as $wt_L(x) = \displaystyle\sum_{i=0}^{s-1}wt_L(x_i)$ and the Lee distance between two elements $x,y\in\mathcal{S}^r$ is defined as $d_L(x,y) = wt_L(x-y)$.\vspace{0.05 cm}

From the above definition of $\overline{\phi}_2$, it can be observed that $\overline{\phi}_2$ is a distance preserving map from $\mathcal{S}^s$ (Lee weight) to $\mathbb{Z}_p^{3s}$ (Hamming weight). Further, if $\mathfrak{C}$ is an $\mathcal{S}$-additive code with parameters $[s,k,d]$ then $\overline{\phi}_2(\mathfrak{C})$ is a $[3s,k,d]$-code over $\mathbb{Z}_p$.\vspace{0.01 cm}

Now, similar to the Theorem \ref{Theorem8} and Corollary \ref{Corollary3}, we have the following two results.
\begin{theorem}\label{Theorem9}
For any two elements $x,x^{\prime}\in\mathcal{S}^s$,$$\overline{\phi}_2(x)\cdot\overline{\phi}_2(x^{\prime}) = \hat{\pi}_1\circ \phi_2(x\cdot x^{\prime}),$$where $\hat{\pi}_1 : \mathbb{Z}_p^3\rightarrow\mathbb{Z}_p$ is the projection map defined as $\hat{\pi}_1(a,b,d) = a$ for $a,b,d\in\mathbb{Z}_p$. In particular, the images of two orthogonal elements in $\mathcal{S}^s$ under the map $\overline{\phi}_2$ are also orthogonal.
\end{theorem}
\begin{proof}
Similar to the proof of Theorem \ref{Theorem8}.
\end{proof}
\begin{corollary}
Let $\mathfrak{C}$ be an $\mathcal{S}$-additive code of length $s$. Then$$\overline{\phi}_2(\mathfrak{C}^{\perp}) = \overline{\phi}_2(\mathfrak{C})^{\perp}.$$
\end{corollary}
\begin{proof}
Similar to the proof of Corollary \ref{Corollary3}.
\end{proof}
Now, we define a Gray map from $\mathbb{Z}_p\mathcal{R}\mathcal{S}$ to $\mathbb{Z}_p^6$ by using the Gray maps $\phi_1$ and $\phi_2$. Define$$\phi : \mathbb{Z}_p\mathcal{R}\mathcal{S}\rightarrow\mathbb{Z}_p^6$$by$$\phi(x,y,z) = (x,\phi_1(y),\phi_2(z))\quad\text{for all }x\in\mathbb{Z}_p,~y\in\mathcal{R},~z\in\mathcal{S}.$$Obviously, the map $\phi$ is also $\mathbb{Z}_p$-linear and bijective. The Lee weight of an element $\alpha = (x,y,z)\in\mathbb{Z}_p\mathcal{R}\mathcal{S}$ is defined as $wt_L(\alpha) = wt_L(x)+wt_L(y)+wt_L(z)$, where $wt_L(x) = \text{min}(x,p-x)$. The Lee distance between two elements $\alpha = (x,y,z),~\alpha^{\prime} = (x^{\prime},y^{\prime},z^{\prime})\in\mathbb{Z}_p\mathcal{R}\mathcal{S}$, is defined as $d_L(\alpha,\alpha^{\prime}) = wt_L(\alpha - \alpha^{\prime})$.\vspace{0.01 cm}

The map $\phi$ can be extended to $\overline{\phi} : \mathbb{Z}_p^q\mathcal{R}^r\mathcal{S}^s\rightarrow\mathbb{Z}_p^{q+2r+3s}$ by$$\overline{\phi}(x,y,z)
= (x,\overline{\phi}_1(y),\overline{\phi}_2(z)),$$where $x\in\mathbb{Z}_p^q,~y\in\mathcal{R}^r,~z\in\mathcal{S}^s$.\vspace{0.01 cm}

The Lee weight of $\alpha = (x_0,x_1,\ldots,x_{q-1}\vert y_0,y_1,\ldots,y_{r-1}\vert z_0,z_1,\ldots,z_{s-1})\in\mathbb{Z}_p^q\mathcal{R}^r\mathcal{S}^s$ is $wt_L(\alpha) = \displaystyle\sum_{i=0}^{q-1} wt_L(x_i)+\displaystyle\sum_{i=0}^{r-1} wt_L(y_i)+\displaystyle\sum_{i=0}^{s-1} wt_L(z_i)$ and the Lee distance between $\alpha, \alpha^{\prime}\in\mathbb{Z}_p^q\mathcal{R}^r\mathcal{S}^s$ is $d_L(\alpha,\alpha^{\prime}) = wt_L(\alpha-\alpha^{\prime})$.\vspace{0.01 cm}

We observe that $\overline{\phi}$ is also a distance preserving map from $\mathbb{Z}_p^q\mathcal{R}^r\mathcal{S}^s$ (Lee weight) to $\mathbb{Z}_p^{q+2r+3s}$ (Hamming weight). Further, if $\mathfrak{C}$ is a $\mathbb{Z}_p^q\mathcal{R}^r\mathcal{S}^s$-additive code of block length $(q,r,s)$, having $p^k$ codewords with minimum distance $d$ then $\overline{\phi}(\mathfrak{C})$ is a $[q+2r+3s,k,d]$-code over $\mathbb{Z}_p$.
\begin{theorem}\label{Theorem10}
For any two elements $\alpha = (x,y,z),~\alpha^{\prime} = (x^{\prime},y^{\prime},z^{\prime})\in\mathbb{Z}_p^q\mathcal{R}^r\mathcal{S}^s$, we have$$\overline{\phi}(\alpha)\cdot\overline{\phi}(\alpha^{\prime}) = x\cdot x^{\prime}+\pi_1\circ\overline{\phi}_1(y\cdot y^{\prime})+\hat{\pi}_1\circ\overline{\phi}_2(z\cdot z^{\prime}),$$where $x,x^{\prime}\in\mathbb{Z}_p^q,~ y,y^{\prime}\in\mathcal{R}^r,~ z,z^{\prime}\in\mathcal{S}^s$.
\end{theorem}
\begin{proof}
We know that $~\overline{\phi}(\alpha) = (x,\overline{\phi}_1(y),\overline{\phi}_2(z))$ and $~\overline{\phi}(\alpha^{\prime}) = (x^{\prime},\overline{\phi}_1(y^{\prime}),\overline{\phi}_2(z^{\prime}))$. Thus
\begin{align*}
    \overline{\phi}(\alpha)\cdot\overline{\phi}(\alpha^{\prime}) =&~ x\cdot x^{\prime}+\overline{\phi}_1(y)\cdot\overline{\phi}_1(y^{\prime})+\overline{\phi}_2(z)\cdot\overline{\phi}_2(z^{\prime})
\end{align*}
Now, the proof follows from Theorem \ref{Theorem8} and Theorem \ref{Theorem9}.
\end{proof}
Considering the inner product which we defined in Definition \ref{Definition6}, we get

from Theorem \ref{Theorem10} that the images of two orthogonal elements in $\mathbb{Z}_p^q\mathcal{R}^r\mathcal{S}^s$ under the map $\overline{\phi}$ are also orthogonal and we have the following corollary.
\begin{corollary}
If $\mathfrak{C}$ is a $\mathbb{Z}_p\mathcal{R}\mathcal{S}$-additive code of block length $(q,r,s)$ then$$\overline{\phi}(\mathfrak{C}^{\perp}) = \overline{\phi}(\mathfrak{C})^{\perp}.$$
\end{corollary}
\begin{proof}
Similar to the proof of Corollary \ref{Corollary3}.
\end{proof}

\subsection{Results on Gray images of additive codes}
Now, we will give a few results related to the Gray images of additive cyclic and additive constacyclic codes.

\begin{lemma}\label{Lemma3}
Let $\sigma_1$, $\theta_2$, $\theta_{\mu_1,2}$ and $T_{\mu_1}$ are respectively the cyclic shift operator, the $2$-quasi-cyclic shift operator, the $(\mu_1,2)$-quasi-twisted shift operator and the $\mu_1$-constacyclic shift operator on $\mathcal{R}^r$ with $\mu_1\in \mathbb{Z}_p^{\ast}$. Then
\begin{enumerate}
    \item $\overline{\phi}_1\circ\sigma_1 = \theta_2\circ\overline{\phi}_1$;
    \item $\overline{\phi}_1\circ T_{\mu_1} = \theta_{\mu_1,2}\circ\overline{\phi}_1$.
\end{enumerate}
\end{lemma}
\begin{proof}
Let $x=(x_0,x_1,\ldots,x_{r-1})\in\mathcal{R}^r$ where $x_i = a_i+ub_i$ with $a_i,b_i\in\mathbb{Z}_p$ for $i = 0,1,\ldots,r-1$.
\begin{enumerate}
    \item We have
    \begin{align*}
    \overline{\phi}_1\circ\sigma_1(x) &=~ \overline{\phi}_1(x_{r-1},x_0,\ldots,x_{r-2})\\
    &=~ (a_{r-1}+b_{r-1},a_0+b_0,\ldots,a_{r-2}+b_{r-2}\vert\kappa b_{r-1},\kappa b_0,\ldots,\kappa b_{r-2}),
    \end{align*}
    and also
    \begin{align*}
        \theta_2\circ\overline{\phi}_1(x) &=~  \theta_2(a_0+b_0,a_1+b_1,\ldots,a_{r-1}+b_{r-1}\vert \kappa b_0,\kappa b_1,\ldots,\kappa b_{r-1})\\
        &=~ (a_{r-1}+b_{r-1},a_0+b_0,\ldots,a_{r-2}+b_{r-2}\vert\kappa b_{r-1},\kappa b_0,\ldots,\kappa b_{r-2}).
    \end{align*}
    Thus, $\overline{\phi}_1\circ\sigma_1 = \theta_2\circ\overline{\phi}_1$.
    \item We have
    \begin{align*}
    \overline{\phi}_1\circ T_{\mu_1}(x) &=~ \overline{\phi}_1(\mu_1x_{r-1},x_0,\ldots,x_{r-2})\\
    &=~ (\mu_1(a_{r-1}+b_{r-1}),a_0+b_0,\ldots,a_{r-2}+b_{r-2}\vert\kappa\mu_1 b_{r-1},\kappa b_0,\ldots,\kappa b_{r-2}),
    \end{align*}
    and
    \begin{align*}
        \theta_{\mu_1,2}\circ\overline{\phi}_1(x) &=~  \theta_{\mu_1,2}(a_0+b_0,a_1+b_1,\ldots,a_{r-1}+b_{r-1}\vert \kappa b_0,\kappa b_1,\ldots,\kappa b_{r-1})\\
        &=~ (\mu_1(a_{r-1}+b_{r-1}),a_0+b_0,\ldots,a_{r-2}+b_{r-2}\vert\kappa\mu_1 b_{r-1},\kappa b_0,\ldots,\kappa b_{r-2}).
    \end{align*}
    Thus, $\overline{\phi}_1\circ T_{\mu_1} = \theta_{\mu_1,2}\circ\overline{\phi}_1$.
\end{enumerate}
\end{proof}
Using the above lemma, we now have the following theorem on the Gray image of cyclic and constacyclic codes over $\mathcal{R}$.

\begin{theorem}
Let $\mu_1\in \mathbb{Z}_p^{\ast}$. Then we have the following:
\begin{enumerate}
    \item If $\mathfrak{C}$ is a cyclic code of length $r$ over $\mathcal{R}$ then $\overline{\phi}_1(\mathfrak{C})$ is a $2$-quasi-cyclic code of length $2r$ over $\mathbb{Z}_p$.
    \item If $\mathfrak{C}$ is a $\mu_1$-constacyclic code of length $r$ over $\mathcal{R}$ then $\overline{\phi}_1(\mathfrak{C})$ is a $(\mu_1,2)$-quasi-twisted code of length $2r$ over $\mathbb{Z}_p$.
\end{enumerate}
\end{theorem}
\begin{proof}
\begin{enumerate}
    \item Suppose $\mathfrak{C}$ is a cyclic code of length $r$ over $\mathcal{R}$. Then $\sigma_1(\mathfrak{C}) = \mathfrak{C}$. Now from Lemma \ref{Lemma3} we have$$\theta_2(\overline{\phi}_1(\mathfrak{C})) = \overline{\phi}_1(\sigma_1(\mathfrak{C})) = \overline{\phi}_1(\mathfrak{C}),$$and this implies that $\overline{\phi}_1(\mathfrak{C})$ is a $2$-quasi-cyclic code of length $2r$ over $\mathbb{Z}_p$.
    \item If $\mathfrak{C}$ is a $\mu_1$-constacyclic code of length $r$ over $\mathcal{R}$ then $T_{\mu_1}(\mathfrak{C}) = \mathfrak{C}$. Now from Lemma \ref{Lemma3} we have$$\theta_{\mu_1,2}(\overline{\phi}_1(\mathfrak{C})) = \overline{\phi}_1(T_{\mu_1}(\mathfrak{C})) = \overline{\phi}_1(\mathfrak{C}),$$which implies that $\overline{\phi}_1(\mathfrak{C})$ is a $(\mu_1,2)$-quasi-twisted code of length $2r$ over $\mathbb{Z}_p$.
\end{enumerate}
\end{proof}

\begin{lemma}\label{Lemma4}
Let $\sigma_2$, $\theta_3$, $\theta_{\mu_2,3}$ and $T_{\mu_2}$ are respectively the cyclic shift operator, the $2$-quasi-cyclic shift operator, the $(\mu_2,3)$-quasi-twisted shift operator and the $\mu_2$-constacyclic shift operator on $\mathcal{S}^s$ with $\mu_2\in \mathbb{Z}_p^{\ast}$. Then
\begin{enumerate}
    \item $\overline{\phi}_2\circ\sigma_2 = \theta_3\circ\overline{\phi}_2$;
    \item $\overline{\phi}_2\circ T_{\mu_2} = \theta_{\mu_2,3}\circ\overline{\phi}_2$.
\end{enumerate}
\end{lemma}
\begin{proof}
Let $x=(x_0,x_1,\ldots,x_{s-1})\in\mathcal{S}^s$ where $x_i = a_i+ub_i+u^2d_i$ with $a_i,b_i,d_i\in\mathbb{Z}_p$ for $i = 0,1,\ldots,s-1$.
\begin{enumerate}
    \item We have
    \begin{align*}
    \overline{\phi}_2\circ\sigma_2(x) =&~ \overline{\phi}_2(x_{s-1},x_0,\ldots,x_{s-2})\\
    =&~ (a_{s-1}+b_{s-1}+d_{s-1},a_0+b_0+d_0,\ldots,a_{s-2}+b_{s-2}+d_{s-2}\vert\kappa (b_{s-1}+d_{s-1}),\\
    &~~\kappa (b_0+d_0),\ldots,\kappa (b_{s-2}+d_{s-2})\vert b_{s-1},b_0,\ldots,b_{s-2}).
    \end{align*}
    Also,
    \begin{align*}
        \theta_3\circ\overline{\phi}_2(x) =&~  \theta_3(a_0+b_0+d_0,a_1+b_1+d_1,\ldots,a_{s-1}+b_{s-1}+d_{s-1}\vert \kappa (b_0+d_0),\\
        &~~~~~\kappa (b_1+d_1),\ldots,\kappa (b_{s-1}+d_{s-1})\vert b_0,b_1,\ldots,b_{s-1})\\
        =&~ (a_{s-1}+b_{s-1}+d_{s-1},a_0+b_0+d_0,\ldots,a_{s-2}+b_{s-2}+d_{s-2}\vert\kappa (b_{s-1}+d_{s-1}),\\
    &~~\kappa (b_0+d_0),\ldots,\kappa (b_{s-2}+d_{s-2})\vert b_{s-1},b_0,\ldots,b_{s-2}).
    \end{align*}
    Thus, $\overline{\phi}_2\circ\sigma_2 = \theta_3\circ\overline{\phi}_2$.
    \item We have
    \begin{align*}
    \overline{\phi}_2\circ T_{\mu_2}(x) =&~ \overline{\phi}_2(\mu_2x_{s-1},x_0,\ldots,x_{s-2})\\
    =&~ (\mu_2(a_{s-1}+b_{s-1}+d_{s-1}),a_0+b_0+d_0,\ldots,a_{s-2}+b_{s-2}+d_{s-2}\vert\kappa\mu_2 (b_{s-1}+d_{s-1}),\\
    &~~~\kappa (b_0+d_0),\ldots,\kappa (b_{s-2}+d_{s-2})\vert \mu_2b_{s-1},b_0,\ldots,b_{s-2}),
    \end{align*}
    and
    \begin{align*}
        \theta_{\mu_2,3}\circ\overline{\phi}_2(x) =&~  \theta_{\mu_2,3}(a_0+b_0+d_0,a_1+b_1+d_1,\ldots,a_{s-1}+b_{s-1}+d_{s-1}\vert \kappa (b_0+d_0),\\
        &~~~~~~~~~\kappa (b_1+d_1),\ldots,\kappa (b_{s-1}+d_{s-1})\vert b_0,b_1,\ldots,b_{s-1})\\
        =&~ (\mu_2(a_{s-1}+b_{s-1}+d_{s-1}),a_0+b_0+d_0,\ldots,a_{s-2}+b_{s-2}+d_{s-2}\vert\kappa\mu_2 (b_{s-1}+d_{s-1}),\\
        &~~~\kappa (b_0+d_0),\ldots,\kappa (b_{s-2}+d_{s-2})\vert \mu_2b_{s-1},b_0,\ldots,b_{s-2}).
    \end{align*}
    Thus, $\overline{\phi}_2\circ T_{\mu_2} = \theta_{\mu_2,3}\circ\overline{\phi}_2$.
\end{enumerate}
\end{proof}
Now, we have the following theorem on the Gray image of cyclic and constacyclic codes over $\mathcal{S}$.

\begin{theorem}
Let $\mu_2\in \mathbb{Z}_p^{\ast}$. Then we have the following:
\begin{enumerate}
    \item If $\mathfrak{C}$ is a cyclic code of length $s$ over $\mathcal{S}$ then $\overline{\phi}_2(\mathfrak{C})$ is a $3$-quasi-cyclic code of length $3s$ over $\mathbb{Z}_p$.
    \item If $\mathfrak{C}$ is a $\mu_2$-constacyclic code of length $s$ over $\mathcal{S}$ then $\overline{\phi}_2(\mathfrak{C})$ is a $(\mu_2,3)$-quasi-twisted code of length $3s$ over $\mathbb{Z}_p$.
\end{enumerate}
\end{theorem}
\begin{proof}
\begin{enumerate}
    \item If $\mathfrak{C}$ is a cyclic code of length $s$ over $\mathcal{S}$ then $\sigma_2(\mathfrak{C}) = \mathfrak{C}$. Now from Lemma \ref{Lemma4} we have$$\theta_3(\overline{\phi}_2(\mathfrak{C})) = \overline{\phi}_2(\sigma_2(\mathfrak{C})) = \overline{\phi}_2(\mathfrak{C}),$$and this implies that $\overline{\phi}_2(\mathfrak{C})$ is a $3$-quasi-cyclic code of length $3s$ over $\mathbb{Z}_p$.
    \item Assume that $\mathfrak{C}$ is a $\mu_2$-constacyclic code of length $s$ over $\mathcal{S}$. Then $T_{\mu_2}(\mathfrak{C}) = \mathfrak{C}$. Now from Lemma \ref{Lemma4} we have$$\theta_{\mu_2,3}(\overline{\phi}_2(\mathfrak{C})) = \overline{\phi}_2(T_{\mu_2}(\mathfrak{C})) = \overline{\phi}_2(\mathfrak{C}),$$which implies that $\overline{\phi}_2(\mathfrak{C})$ is a $(\mu_2,3)$-quasi-twisted code of length $3s$ over $\mathbb{Z}_p$.
\end{enumerate}
\end{proof}
The next result is on the Gray image of a constacyclic code over $\mathbb{Z}_p\mathcal{R}\mathcal{S}$.
\begin{theorem}
Let $\mu_0,\mu_1,\mu_2\in\mathbb{Z}_p^\ast$ and suppose $\mathfrak{C}$ is a $(\mu_0,\mu_1,\mu_2)$-constacyclic code of block length $(q,r,s)$ over $\mathbb{Z}_p\mathcal{R}\mathcal{S}$. Then $\overline{\phi}(\mathfrak{C})$ is a generalized $(\mu_0,\mu_1,\mu_1,\mu_2,\mu_2,\mu_2)$-quasi-twisted code of block length $(s,r,r,q,q,q)$.
\end{theorem}
\begin{proof}
Take $w\in\overline{\phi}(\mathfrak{C})$. Then there exists $v\in\mathfrak{C}$ such that $w = \overline{\phi}(v)$. Let $$v = (x_0,x_1,\ldots,x_{q-1}\vert y_0,y_1,\ldots,y_{r-1}\vert z_0,z_1,\ldots,z_{s-1}),$$where $y_i = a_i+ub_i$, $z_i = a_i^{\prime}+ub_i^{\prime}+u^2d_i^{\prime}$ and $x_i,a_i,b_i,a_i^{\prime},b_i^{\prime},d_i^{\prime}\in\mathbb{Z}_p$. Then
\begin{align*}
    \overline{\phi}(v) =&~ (x_0,x_1,\ldots,x_{q-1}\vert a_0+b_0,a_1+b_1,\ldots,a_{r-1}+b_{r-1}\vert \kappa b_0,\kappa b_1,\ldots,\kappa b_{r-1}\vert a_0^{\prime}+b_0^{\prime}+d_0^{\prime},a_1^{\prime}+b_1^{\prime}+\\
    &~~~d_1^{\prime},\ldots,a_{s-1}^{\prime}+b_{s-1}^{\prime}d_{s-1}^{\prime}\vert \kappa (b_0^{\prime}+d_0^{\prime}),\kappa (b_1^{\prime}+d_1^{\prime}),\ldots,\kappa (b_{s-1}^{\prime}+d_{s-1}^{\prime})\vert b_0^{\prime},b_1^{\prime},\ldots,b_{s-1}^{\prime}).
\end{align*}
Let $T_{\mu_0,\mu_1,\mu_2}$ be the $(\mu_0,\mu_1,\mu_2)$-constacyclic shift operator on $\mathbb{Z}_p^q\mathcal{R}^r\mathcal{S}^s$. Then we have
\begin{align*}
    \overline{\phi}\circ T_{\mu_0,\mu_1,\mu_2}(v) =&~ \overline{\phi}(\mu_0x_{q-1},x_0,\ldots,x_{q-2}\vert \mu_1y_{r-1},y_0,\ldots,y_{r-2}\vert \mu_2z_{s-2},z_0,\ldots,z_{s-2})\\
    =&~ (\mu_0x_{q-1},x_0,\ldots,x_{q-2}\vert \mu_1(a_{r-1}+b_{r-1}),a_0+b_0,\ldots,a_{r-2}+b_{r-2}\vert \kappa\mu_1 b_{r-1},\\
    &~~~\kappa b_0,\ldots,\kappa b_{r-2}\vert \mu_2(a_{s-1}^{\prime}+b_{s-1}^{\prime}+d_{s-1}^{\prime}),a_0^{\prime}+b_0^{\prime}+d_0^{\prime},\ldots,a_{s-2}^{\prime}+b_{s-2}^{\prime}+\\
    &~~~d_{s-2}^{\prime}\vert \kappa\mu_2(b_{s-1}^{\prime}+d_{s-1}^{\prime}),\kappa (b_0^{\prime}+d_0^{\prime}),\ldots,\kappa(b_{s-2}^{\prime}+d_{s-2}^{\prime})\vert \mu_2b_{s-1}^{\prime},b_0^{\prime},\ldots,b_{s-2}^{\prime}).
\end{align*}
Since $\mathfrak{C}$ is a $(\mu_0,\mu_1,\mu_2)$-constacyclic code, $\overline{\phi}\circ T_{\mu_0,\mu_1,\mu_2}(v)\in\overline{\phi}(\mathfrak{C})$ and we observe that $\overline{\phi}\circ T_{\mu_0,\mu_1,\mu_2}(v)$ is the generalized $(\mu_0,\mu_1,\mu_1,\mu_2,\mu_2,\mu_2)$-quasi-twisted shift of $\overline{\phi}(v)$. Therefore, $\overline{\phi}(\mathfrak{C})$ is a generalized $(\mu_0,\mu_1,\mu_1,\mu_2,\mu_2,\mu_2)$-quasi-twisted code of block length $(s,r,r,q,q,q)$.
\end{proof}

\section{The weight enumerator and MacWilliams identities}
In this section, we study different weight enumerators, such as complete weight enumerator, symmetrized weight enumerator, etc., and establish the MacWilliams identities.
\subsection{The complete weight enumerator and the Hamming weight enumerator}
First, we arrange the members of each of $\mathbb{Z}_p$, $\mathcal{R}\text{ and }\mathcal{S}$ in a particular order as follows:
\begin{align*}
    \mathbb{Z}_p =&~ \{0,1,2,\ldots,p-1\};\\
    \mathcal{R} =&~ \{0,1,u,1+u,2,2u,2+u,1+2u,2+2u.\ldots,(p-1),(p-1)u,(p-1)+u,1+(p-1)u,\\
    &~~~(p-1)+2u,2+(p-1)u,\ldots,(p-1)+(p-1)u\};\\
    \mathcal{S} =&~ \{0,1,u,1+u,u^2,1+u^2,u+u^2,1+u+u^2,2,2u,2+u,1+2u,2+2u,2u^2, 2+u^2,1+2u^2,2+\\
    &~~~2u^2,2u+u^2,u+2u^2,2u+2u^2,2+u+u^2,1+2u+u^2,2+2u+u^2,1+u+2u^2,2+u+2u^2,\\
    &~~~1+2u+2u^2,2+2u+2u^2,3,3u,3+u,1+3u,3+2u,2+3u,3+3u,3u^2,3+u^2,1+3u^2,3+\\
    &~~~2u^2,2+3u^2,3+3u^2,3u+u^2,u+3u^2,3u+2u^2,2u+3u^2,3u+3u^2,3+u+u^2,1+3u+u^2,\\
    &~~~3+2u+u^2,2+3u+u^2,3+3u+u^2,1+u+3u^2,3+u+2u^2,2+u+3u^2,3+u+3u^2,1+\\
    &~~~3u+2u^2,1+2u+3u^2,1+3u+3u^2,3+2u+2u^2,2+3u+2u^2,3+3u+2u^2,2+2u+3u^2,\\
    &~~~3+2u+3u^2,2+3u+3u^2,3+3u+3u^2,\ldots,(p-1)+(p-1)u+(p-1)u^2\}
\end{align*}
Using the above ordering of the members of each of $\mathbb{Z}_p$, $\mathcal{R}\text{ and }\mathcal{S}$, now we can order the members of $\mathbb{Z}_p\mathcal{R}\mathcal{S}$ in the following way.
\begin{align*}
   \mathbb{Z}_p\mathcal{R}\mathcal{S} =&~ \{f_1,f_2,\ldots,f_{p^6}\}\\
   =&~\{(0,0,0),(0,0,1),(0,0,u),(0,0,1+u),(0,0,u^2),(0,0,1+u^2),(0,0,u+u^2),(0,0,1+u+u^2),\\
   &~~(0,0,2),(0,0,2u),(0,0,2+u),(0,0,1+2u),(0,0,2+2u),(0,0,2u^2),(0,0,2+u^2),(0,0,1+\\
   &~~2u^2),(0,0,2+2u^2),(0,0,2u+u^2),(0,0,u+2u^2),(0,0,2u+2u^2),(0,0,2+u+u^2),(0,0,1+\\
   &~~2u+u^2),(0,0,2+2u+u^2),(0,0,1+u+2u^2),(0,0,2+u+2u^2),(0,0,1+2u+2u^2),(0,0,\\
   &~~2+2u+2u^2),(0,0,3),\ldots,(0,0,(p-1)+(p-1)u+(p-1)u^2),(0,1,0),(0,1,1),(0,1,u),(0,\\
   &~~1,1+u),\ldots,(0,1,(p-1)+(p-1)u+(p-1)u^2),(0,u,0),(0,u,1),(0,u,u),\ldots,(0,u,(p-1)+\\
   &~~(p-1)u+(p-1)u^2),(0,1+u,0),(0,1+u,1),(0,1+u,u),\ldots,(0,1+u,(p-1)+(p-1)u+\\
   &~~(p-1)u^2),\ldots,(0,(p-1)+(p-1)u,(p-1)+(p-1)u+(p-1)u^2),(1,0,0),(1,0,1),(1,0,u),\\
   &~~\ldots,((p-1),(p-1)+(p-1)u,(p-1)+(p-1)u+(p-1)u^2)\}
\end{align*}
Now, we are ready to define the complete weight enumerator of a $\mathbb{Z}_p\mathcal{R}\mathcal{S}$-additive code $\mathfrak{C}$.

\begin{definition}\label{Definition14}
 Let $\mathfrak{C}$ be a $\mathbb{Z}_p\mathcal{R}\mathcal{S}$-additive code of length $n$. Then the complete weight enumerator of $\mathfrak{C}$ is denoted by $\mathcal{W}_{\mathbf{C}}^{(\mathfrak{C})}$ and is defined by
$$\mathcal{W}_{\mathbf{C}}^{(\mathfrak{C})}(x_1,x_2,\ldots,x_{p^6}) = \sum_{c\in\mathfrak{C}}\prod_{i=1}^{p^6} x_i^{w_{f_i}(c)}~,$$where for each $c = (c_0,c_1,\ldots,c_{n-1})\in\mathfrak{C}$,  $w_{f_i}(c) = |\{j\mid c_j = f_i, 0\leq j\leq n-1\}|$ for $i = 1,2,\ldots,p^6$.
\end{definition}
From the above definition, it is evident that $\mathcal{W}_{\mathbf{C}}^{(\mathfrak{C})}(x_1,x_2,\ldots,x_{p^6})$ is a homogeneous polynomial. The total degree of each monomial in $\mathcal{W}_{\mathbf{C}}^{(\mathfrak{C})}(x_1,x_2,\ldots,x_{p^6})$ is $n$. Also, we observe that $\mathcal{W}_{\mathbf{C}}^{(\mathfrak{C})}(1,1,\ldots,1) = |\mathfrak{C}|$.\vspace{0.01 cm}

Now, we will investigate the MacWilliams identity corresponding to the complete weight enumerator. First, we define a generating character on $\mathbb{Z}_p\mathcal{R}\mathcal{S}$.
\begin{definition}
 Define $\chi: \mathbb{Z}_p\mathcal{R}\mathcal{S}\rightarrow\mathbb{C}^{\ast}$ by$$\chi(a,a^{\prime}+ub^{\prime},a^{\prime\prime}+ub^{\prime\prime}+u^2d^{\prime\prime}) = (-1)^{a+a^{\prime}+b^{\prime}+a^{\prime\prime}+b^{\prime\prime}+d^{\prime\prime}}.$$It can be easily verified that $\chi$-image of any non-zero ideal is always non-trivial and hence by Lemma \ref{Lemma5}, $\chi$ is a generating character on $\mathbb{Z}_p\mathcal{R}\mathcal{S}$.
\end{definition}
Suppose, $P = [p_{ij}]$ is a matrix of order $p^6$ with $p_{ij} = \chi(f_if_j)$ where $\chi$ is the generating character defined above and $f_i,f_j\in\mathbb{Z}_p\mathcal{R}\mathcal{S}$ for all $i,j = 1,2,\ldots,p^6$. We find the matrix $P$ for two different cases, $p=2$ and $p\neq 2$.\vspace{0.01 cm}

First, we consider the case when $p= 2$. Here, we have
\begin{align*}
P =
\begin{bmatrix}
    B & ~~~~~B & ~~~~~B & ~~~~~B & \dots & ~~B\\
    B &~ -B &~ ~~~~B & ~-B & \dots & -B\\
    B &~ ~~~~B & ~~~~~B &~~~~~B & \dots & ~~B\\
    B &~ -B &~~~~~B &~ -B & \dots & -B\\
    \vdots &~~~~~\vdots &~~~~~\vdots &~~~~~\vdots &~\ddots & ~~~\vdots\\
    B & ~-B &~ ~~~~B &~ -B & \dots & -B
\end{bmatrix},
\end{align*}
where
\begin{align*}
B =
\begin{bmatrix}
    A & ~~~~~A & ~~~~~A & ~~~~~A & ~~~~~A & \dots & ~~A\\
    A &~ -A &~ -A & ~~~~~A & ~~~~~A & \dots & ~~A\\
    A &~ -A & ~~~~~A &~ -A & ~~~~~A & \dots & -A\\
    A & ~~~~~A &~ -A &~ -A & ~~~~~A & \dots & -A\\
    A & ~~~~~A & ~~~~~A & ~~~~~A & ~~~~~A & \dots & ~~A\\
    ~\vdots &~~~~~~\vdots &~~~~~~\vdots &~~~~~~\vdots &~~~~~~\vdots &~~\ddots & ~~~\vdots\\
    A & ~~~~~A &~ -A &~ -A & ~~~~~A & \dots & -A
\end{bmatrix}
\end{align*}
be a matrix of order $p^5$ and $A$ be a matrix of order $p^3$, given by
\begin{align*}
A =
\begin{bmatrix}
    1 & ~~~~~1 & ~~~~~1 & ~~~~~1 & ~~~~~1 &~~~~~1 &~~~~~1 &~~~~~1 & \dots & ~~1\\
    1 &~ -1 &~ -1 & ~~~~~1 &~ -1  &~~~~~1 &~~~~~1 &~-1 & \dots & -1\\
    1 &~ -1 &~ -1 & ~~~~~1 & ~~~~~1 &~-1 &~-1 &~~~~~1 & \dots & ~~1\\
    1 & ~~~~~1 & ~~~~~1 & ~~~~~1 &~-1 &~-1 &~-1 &~-1 & \dots & -1\\
    1 & ~-1 & ~~~~~1 &~ -1 & ~~~~~1 &~-1 &~~~~~1 &~-1 & \dots & -1\\
    1 & ~~~~~1 &~ -1 &~ -1 &~ -1 &~-1 &~~~~~1 &~~~~~1 & \dots & ~~1\\
    1 & ~~~~~1 &~ -1 &~ -1 & ~~~~~1 &~~~~~1 &~-1 &~-1 & \dots & -1\\
    1 &~ -1 & ~~~~~1 &~ -1 &~ -1 &~~~~~1 &~-1 &~~~~~1 & \dots & ~~1\\
    \vdots &~~~~~\vdots &~~~~~\vdots &~~~~~\vdots &~~~~~\vdots &~~~~~\vdots &~~~~~\vdots &~~~~~\vdots &~~\ddots & ~~\vdots\\
    1 &~ -1 & ~~~~~1 &~ -1 &~ -1 &~~~~~1 &~-1 &~~~~~1 & \dots & ~~1
\end{bmatrix}.
\end{align*}
The matrix $A$ is actually the submatrix of $P$, consisting of the first $p^3$ rows and the first $p^3$ columns. Similarly, the matrix $B$ is also a submatrix of $P$, consisting of the first $p^5$ rows and the first $p^5$ columns. In other words, the matrix $A$ is generated by the first $p^3$ members of $\mathbb{Z}_p\mathcal{R}\mathcal{S}$ whereas the matrix $B$ is generated by the first $p^5$ members of $\mathbb{Z}_p\mathcal{R}\mathcal{S}$. Now, for $p\neq 2$, we have
\begin{align*}
P =
\begin{bmatrix}
    B & ~~~~~B & ~~~~~B & ~~~~~B & \dots & ~~B\\
    B &~ -B & ~~~~~B & ~-B & \dots & ~~B\\
    B & ~~~~~B & ~~~~~B &~~~~~B & \dots & ~~B\\
    B &~ -B &~~~~~B  &~ -B & \dots & ~~B\\
    \vdots &~~~~~\vdots &~~~~~\vdots &~~~~~\vdots &~\ddots & ~~~\vdots\\
    B & ~~~~~B &~ ~~~~B &~ ~~~~B & \dots & -B
\end{bmatrix},
\end{align*}
where
\begin{align*}
B =
\begin{bmatrix}
    A & ~~~~~A & ~~~~~A & ~~~~~A & ~~~~~A & \dots & ~~A\\
    A &~ -A &~ -A & ~~~~~A & ~~~~~A & \dots & ~~A\\
    A &~ -A & ~~~~~A &~ -A & ~~~~~A & \dots & ~~A\\
    A & ~~~~~A &~ -A &~ -A & ~~~~~A & \dots & ~~A\\
    A & ~~~~~A & ~~~~~A & ~~~~~A & ~~~~~A & \dots & ~~A\\
    ~\vdots &~~~~~~\vdots &~~~~~~\vdots &~~~~~~\vdots &~~~~~~\vdots &~~\ddots & ~~~\vdots\\
    A & ~~~~~A &~ ~~~~A &~ ~~~~A & ~~~~~A & \dots & -A
\end{bmatrix}, ~
A =
\begin{bmatrix}
    1 & ~~~~~1 & ~~~~~1 & ~~~~~1 & ~~~~~1 &~~~~~1 &~~~~~1 &~~~~~1 & \dots & ~~1\\
    1 &~ -1 &~ -1 & ~~~~~1 &~ -1 &~~~~~1 &~~~~~1 &~-1 & \dots & ~~1\\
    1 &~ -1 &~ -1 & ~~~~~1 & ~~~~~1 &~-1 &~-1 &~~~~~1 & \dots & ~~1\\
    1 & ~~~~~1 & ~~~~~1  & ~~~~~1 &~ -1 &~-1 &~-1 &~-1 & \dots & ~~1\\
    1 &~ -1 & ~~~~~1 &~-1 & ~~~~~1  &~-1 &~~~~~1 &~-1 & \dots & ~~1\\
    1 & ~~~~~1 &~ -1 &~ -1 &~ -1 &~-1 &~~~~~1 &~~~~~1 & \dots & ~~1\\
    1 & ~~~~~1 &~ -1 &~ -1 & ~~~~~1  &~~~~~1 &~-1 &~-1 & \dots & ~~1\\
    1 &~ -1 & ~~~~~1 &~ -1 &~ -1 &~~~~~1 &~-1 &~~~~~1 & \dots & ~~1\\
    \vdots &~~~~~\vdots &~~~~~\vdots &~~~~~\vdots &~~~~~\vdots &~~~~~\vdots &~~~~~\vdots &~~~~~\vdots &~~\ddots & ~~\vdots\\
    1 &~ ~~~~1 & ~~~~~1 &~ ~~~~1 &~ ~~~~1 &~~~~~1 &~~~~~1 &~~~~~1 & \dots & ~~1
\end{bmatrix}.
\end{align*}
Now, we find the MacWilliams identity with respect to the complete weight enumerator.

\begin{theorem}\label{Theorem14}
If $\mathfrak{C}$ is a $\mathbb{Z}_p\mathcal{R}\mathcal{S}$-additive code of length $n$ then$$\mathcal{W}_{\mathbf{C}}^{(\mathfrak{C}^{\perp})}(x_1,x_2,\ldots,x_{p^6}) = \frac{1}{|\mathfrak{C}|}\mathcal{W}_{\mathbf{C}}^{(\mathfrak{C})}(P\cdot(x_1,x_2,\ldots,x_{p^6})^T),$$ where $(x_1,x_2,\ldots,x_{p^6})^T$ denotes the transpose of $(x_1,x_2,\ldots,x_{p^6})$.
\end{theorem}
\begin{proof}
The proof follows from \cite{Wood}.
\end{proof}

\begin{example}\label{Example2}
Consider the $\mathbb{Z}_2\mathcal{R}\mathcal{S}$-additive code of length $2$,
\begin{align*}
    \mathfrak{C}_2 =&~ \langle\{(1,0,1+u^2;0,u,0),(0,1+u,0;1,0,1+u)\}\rangle.
\end{align*}
Then $\mathfrak{C}_2$ is a linear code over $\mathbb{Z}_2$ of dimension $6$ and
\begin{align*}
B_1 =&~ \{(1,0,1+u^2;0,u,0),(0,0,u;0,0,0),(0,0,u^2;0,0,0),(0,1+u,0;1,0,1+u),(0,u,0;0,0,u+u^2),\\
&~~~ (0,0,0;0,0,u^2)\}
\end{align*}
forms a $\mathbb{Z}_2$-basis of $\mathfrak{C}_2$. Thus the dual code $\mathfrak{C}_2^{\
\perp}$ is also a linear code over $\mathbb{Z}_2$ of dimension $6$ and
\begin{align*}
B_2 =&~ \{(0,0,0;0,u,0),(1,0,0;0,1,0),(0,u,0;0,0,u^2),(0,u,0;1,0,0),(0,1+u,0;0,0,u+u^2),\\
&~~~ (1,0,u^2;0,0,0)\}
\end{align*}
is a $\mathbb{Z}_2$-basis of $\mathfrak{C}_2^{\perp}$. Now, according to Definition \ref{Definition14}, the complete weight enumerator of $\mathfrak{C}_2$ is given by
\begin{align*}
    \mathcal{W}_{\mathbf{C}}^{(\mathfrak{C}_2)}(x_1,x_2,\ldots,x_{64}) =&~ x_1^2 + x_{38}x_{17} + x_3x_1 + x_{40}x_{17} + 2x_4x_1 + x_{34}x_{17} + x_7x_1 + x_{37}x_{17} + x_{25}x_{37} +\\
    &~ x_{62}x_{53} + x_{27}x_{37} + x_{64}x_{53} + x_{28}x_{37} + x_{58}x_{53} + x_{31}x_{37} + x_{61}x_{53} + x_{17}x_{7} +\\
    &~ x_{54}x_{23} + x_{19}x_{7} + x_{56}x_{23} + x_{20}x_{7}+ x_{50}x_{23}+ x_{23}x_{7}+ x_{53}x_{23} + x_{9}x_{38} +\\
    &~ x_{46}x_{54} + x_{11}x_{38} + x_{48}x_{54} + x_{12}x_{38} + x_{42}x_{54} + x_{15}x_{38} + x_{45}x_{54} + x_{38}x_{20} +\\
    &~ x_{3}x_{4} + x_{40}x_{20} + x_{4}^2 + x_{34}x_{20} + x_{7}x_{4} + x_{37}x_{20} + x_{25}x_{40} + x_{62}x_{56} + x_{27}x_{40} +\\
    &~ x_{64}x_{56} + x_{28}x_{40} + x_{58}x_{56} + x_{31}x_{40} + x_{61}x_{56} + x_{17}x_{3} + x_{54}x_{19} + x_{19}x_{3} +\\
    &~ x_{56}x_{19} + x_{20}x_{3} + x_{50}x_{19} + x_{23}x_{3} + x_{53}x_{19} + x_{9}x_{34} + x_{46}x_{50} + x_{11}x_{34} +\\
    &~ x_{48}x_{50} + x_{12}x_{34} + x_{42}x_{50} + x_{15}x_{34} + x_{45}x_{50}.
\end{align*}
Similarly, the complete weight enumerator of $\mathfrak{C}_2^{\perp}$ is given by
\begin{align*}
    \mathcal{W}_{\mathbf{C}}^{(\mathfrak{C}_2^{\perp})}(x_1,x_2,\ldots,x_{64}) =&~ x_1^2 + 2x_{36}x_{1} + x_{25}x_7 + x_{60}x_{7} + x_{17}x_{33} + x_{52}x_{33} + x_9x_{39} + x_{44}x_{39} + x_{17}x_{4} +\\
    &~ x_{52}x_{4} + x_{9}x_{3} + x_{44}x_{3} + x_{36}^2 + x_{25}x_{35} + x_{60}x_{35} + x_{33}x_{9} + x_{4}x_{9} + x_{57}x_{7} +\\
    &~ x_{28}x_{15} + x_{49}x_{41}+ x_{20}x_{41}+ x_{41}x_{47}+ x_{12}x_{47} + x_{49}x_{12} + x_{20}x_{12} + x_{41}x_{11} +\\
    &~ x_{12}x_{11} + x_{33}x_{44} + x_{4}x_{44} + x_{57}x_{43} + x_{28}x_{43} + x_{1}x_{17} + x_{36}x_{17} + x_{25}x_{23} +\\
    &~ x_{60}x_{23} + x_{17}x_{49} + x_{52}x_{49} + x_{9}x_{55} + x_{44}x_{55} + x_{17}x_{20} + x_{52}x_{20} + x_{9}x_{19} +\\
    &~ x_{44}x_{19} + x_{1}x_{52} + x_{36}x_{52} + x_{25}x_{51} + x_{60}x_{51} + x_{33}x_{25} + x_{4}x_{25} + x_{57}x_{31} +\\
    &~ x_{28}x_{31} + x_{49}x_{57} + x_{20}x_{57} + x_{41}x_{63} + x_{12}x_{63} + x_{49}x_{28} + x_{20}x_{28} + x_{41}x_{27} +\\
    &~ x_{12}x_{27} + x_{33}x_{60} + x_{4}x_{60} + x_{57}x_{59} + x_{28}x_{59}.
\end{align*}
\end{example}
Next, we define the Hamming weight enumerator and find the corresponding MacWilliams identity.

\begin{definition}
 Let $\mathfrak{C}$ be a $\mathbb{Z}_p\mathcal{R}\mathcal{S}$-additive code of length $n$. Then the Hamming weight enumerator of $\mathfrak{C}$ is denoted by $\mathcal{W}_{\mathbf{H}}^{(\mathfrak{C})}$ and is defined by$$\mathcal{W}_{\mathbf{H}}^{(\mathfrak{C})}(x,y) = \sum_{c\in\mathfrak{C}}x^{n-wt_H(c)}y^{wt_H(c)}~,$$where $wt_H(c)$ is the Hamming weight of $(c)$.
\end{definition}
Similar to the complete weight enumerator, the Hamming weight enumerator is also a homogeneous polynomial of degree $n$. Further, we observe that$$\mathcal{W}_{\mathbf{H}}^{(\mathfrak{C})}(x,y) = \mathcal{W}_{\mathbf{C}}^{(\mathfrak{C})}(x,y,y,\ldots,y),$$which gives us a relation between the complete weight enumerator and the Hamming weight enumerator.

\begin{theorem}\label{Theorem15}
If $\mathfrak{C}$ is a $\mathbb{Z}_p\mathcal{R}\mathcal{S}$-additive code of length $n$ then$$\mathcal{W}_{\mathbf{H}}^{(\mathfrak{C}^{\perp})}(x,y) = \frac{1}{|\mathfrak{C}|}\mathcal{W}_{\mathbf{H}}^{(\mathfrak{C})}(x+(p^6-1)y,x-y).$$
\end{theorem}
\begin{proof}
The proof follows from Theorem \ref{Theorem14} and uses the relation between the complete weight enumerator and the Hamming weight enumerator.
\end{proof}

\begin{example}
We consider the same code $\mathfrak{C}_2$ as in Example \ref{Example2}. Then by using the complete weight enumerator of $\mathfrak{C}_2$, found in Example \ref{Example2}, the Hamming weight enumerator of $\mathfrak{C}_2$ is
\begin{align*}
   \mathcal{W}_{\mathbf{H}}^{(\mathfrak{C}_2)}(x,y) =&~ x^2 + 4xy + 59y^2,
\end{align*}
and by using Theorem \ref{Theorem15}, the Hamming weight enumerator of $\mathfrak{C}_2^{\perp}$ is given by
\begin{align*}
   \mathcal{W}_{\mathbf{H}}^{(\mathfrak{C}_2^{\perp})}(x,y) =&~ \frac{1}{64}\mathcal{W}_{\mathbf{H}}^{(\mathfrak{C}_2)}(x+63y,x-y)\\
   =&~ \frac{1}{64}[(x+63y)^2+4(x+63y)(x-y)+59(x-y)^2]\\
   =&~ x^2 + 4xy + 59y^2.
\end{align*}
\end{example}

\subsection{The symmetrized weight enumerator and the Lee weight enumerator}
First, we find the value of $wt_L(f_i)$, the Lee weight of $f_i,~ i = 1,2,\ldots,p^6$, where $f_i$ are the members of $\mathbb{Z}_p\mathcal{R}\mathcal{S}$ in the same order as we have considered earlier.
\begin{align*}
    & wt_L(f_1) = wt_L((0,0,0)) = 0,~ wt_L(f_2) = wt_L((0,0,1)) = 1,~ wt_L(f_3) = wt_L((0,0,u)) = 3,~ wt_L(f_4) =\\
    & wt_L((0,0,1+u)) = 2\text{ for } p=2 \text{ and, } 3 \text{ for } p\neq2,\ldots,~ wt_L(f_{p^3}) = wt_L((0,0,(p-1)+(p-1)u+\\
    & (p-1)u^2)) = 2\text{ for } p=2\text{ and, } 3 \text{ for } p\neq 2,~wt_L(f_{p^3+1}) = wt_L((0,1,0)) = 1, wt_L(f_{p^3+2}) = wt_L((0,1,1))\\
    & = 2,~wt_L(f_{p^3+3}) = wt_L((0,1,u)) = 4,~wt_L(f_{p^3+4}) = wt_L((0,1,1+u)) = 3\text{ for } p=2 \text{ and, } 4 \text{ for }p\neq2,\\
    & \ldots,~ wt_L(f_{2p^3}) = wt_L((0,1,(p-1)+(p-1)u+(p-1)u^2)) = 3 \text{ for } p=2 \text{ and, } 4 \text{ for } p\neq 2,\\
    & ~wt_L(f_{2p^3+1}) = wt_L((0,u,0)) = 2,~wt_L(f_{2p^3+2}) = wt_L((0,u,1)) = 3, ~wt_L(f_{2p^3+3}) = wt_L((0,u,u)) = 5,\\
    & ~wt_L(f_{2p^3+4}) = wt_L((0,u,1+u)) = 4\text{ for } p=2 \text{ and, } 5 \text{ for } p\neq2,\ldots,~ wt_L(f_{p^5}) = wt_L((0,(p-1)+\\
    & (p-1)u,(p-1)+(p-1)u+(p-1)u^2)) = 3 \text{ for } p=2, \text{ and, } 5 \text{ for } p\neq 2, ~wt_L(f_{p^5+1}) = wt_L((1,0,0))\\
    & =1,~wt_L(f_{p^5+2}) = wt_L((1,0,1)) = 2, wt_L(f_{p^5+3}) = wt_L((1,0,u)) = 4, ~wt_L(f_{p^5+4}) = wt_L((1,0,1+u))\\
    &  = 3\text{ for } p=2 \text{ and, } 4 \text{ for } p\neq2,\ldots, wt_L(f_{p^6}) = wt_L((p-1,(p-1)+(p-1)u,(p-1)+(p-1)u+\\
    & (p-1)u^2)) = 4 \text{ for } p=2, \text{ and, } 6 \text{ for other values of } p.
\end{align*}
Note that the Lee weights of the elements of $\mathbb{Z}_p\mathcal{R}\mathcal{S}$ can vary from $0$ to $6$.\vspace{0.01 cm}

Now, we define the symmetrized weight enumerator.

\begin{definition}\label{Definition16}
 Let $\mathfrak{C}$ be a $\mathbb{Z}_p\mathcal{R}\mathcal{S}$-additive code of length $n$. Then the symmetrized weight enumerator of $\mathfrak{C}$ is denoted by $\mathcal{W}_{\mathbf{S}}^{(\mathfrak{C})}$ and is defined by
 \begin{align*}
 \mathcal{W}_{\mathbf{S}}^{(\mathfrak{C})}(W_0,W_1,W_2,\ldots,W_6) =&~ \mathcal{W}_{\mathbf{C}}^{(\mathfrak{C})}(W_0,W_1,W_3,W_{\ast}^{(1)},\ldots,W_{\ast}^{(2)},W_1,W_2,W_4,W_{\ast}^{(3)},\ldots,W_{\ast}^{(4)},W_2,\\
 &~~~~~~~~~W_3,W_5,W_{\ast}^{(5)},\ldots,W_{\ast}^{(6)},W_1,W_2,W_4,W_{\ast}^{(3)},\ldots,W_{\ast}^{(7)}),
 \end{align*}
 where $W_i$ is the variable corresponding to the Lee weight $i$ of the elements of $\mathbb{Z}_p\mathcal{R}\mathcal{S}$ with
 \begin{align*}
 & W_{\ast}^{(1)} =
 \begin{cases}
 W_2 \quad\text{for } p=2\\
 W_3 \quad\text{for } p\neq 2
 \end{cases},~
 W_{\ast}^{(2)} =
 \begin{cases}
 W_2 \quad\text{for } p=2\\
 W_3 \quad\text{for } p\neq 2
 \end{cases},~
 W_{\ast}^{(3)} =
 \begin{cases}
 W_3 \quad\text{for } p=2\\
 W_4 \quad\text{for } p\neq 2
 \end{cases},\\
 & W_{\ast}^{(4)} =
 \begin{cases}
 W_3 \quad\text{for } p=2\\
 W_4 \quad\text{for } p\neq 2
 \end{cases},~
 W_{\ast}^{(5)} =
 \begin{cases}
 W_4 \quad\text{for } p=2\\
 W_5 \quad\text{for } p\neq 2
 \end{cases},~\\
 & W_{\ast}^{(6)} =
 \begin{cases}
 W_3 \quad\text{for } p=2\\
 W_5 \quad\text{for } p\neq 2
 \end{cases}\quad\text{ and }\quad
  W_{\ast}^{(7)} =
 \begin{cases}
 W_4 \quad\text{for } p=2\\
 W_6 \quad\text{for } p\neq 2
 \end{cases}.
  \end{align*}
 \end{definition}
From the above definition, we have
\begin{equation}\label{equation1}
\mathcal{W}_{\mathbf{S}}^{(\mathfrak{C})}(W_0,W_1,W_2,\ldots,W_6) = \sum_{c\in\mathfrak{C}} W_0^{n_0(c)} W_1^{n_1(c)} W_2^{n_2(c)}\cdots W_6^{n_6(c)}~,
\end{equation}
where for each $c = (c_0,c_1,\ldots,c_{n-1})\in\mathfrak{C}$, $n_i(c) = |\{j\mid wt_L(c_j) = i,~ 0\leq j\leq n-1\}|$.\vspace{0.01 cm}

Using the same notations as used in Definition \ref{Definition16}, we find the MacWilliams identity for the symmetrized weight enumerator.

\begin{theorem}\label{Theorem16}
If $\mathfrak{C}$ is a $\mathbb{Z}_p\mathcal{R}\mathcal{S}$-additive code of length $n,$ then$$\mathcal{W}_{\mathbf{S}}^{(\mathfrak{C}^{\perp})}(W_0,W_1,W_2,\ldots,W_6) = \frac{1}{|\mathfrak{C}|}\mathcal{W}_{\mathbf{S}}^{(\mathfrak{C})}(Q\cdot(W_0,W_1,W_2,\ldots,W_6)^T),$$where $Q$ is a matrix of order $7$, and it is the coefficient matrix of the distinct non-zero rows of the matrix
\begin{align*}
    P\cdot & (W_0,W_1,W_3,W_{\ast}^{(1)},\ldots,W_{\ast}^{(2)},W_1,W_2,W_4,W_{\ast}^{(3)},\ldots,W_{\ast}^{(4)},W_2,W_3,W_5,W_{\ast}^{(5)},\ldots,W_{\ast}^{(6)},W_1,W_2,\\
    &~~~W_4,W_{\ast}^{(3)},\ldots,W_{\ast}^{(7)})^T.
\end{align*}
\end{theorem}
\begin{proof}
The proof follows from Theorem \ref{Theorem14} and Definition \ref{Definition16}.
\end{proof}
\begin{example}\label{Example4}
For the code $\mathfrak{C}_2$, defined in Example \ref{Example2}, the symmetrized weight enumerator of $\mathfrak{C}_2$ is given by
\begin{align*}
   \mathcal{W}_{\mathbf{S}}^{(\mathfrak{C}_2)}(W_0,W_1,W_2,\ldots,W_6) =&~ W_0^2 + 5W_2^2 + W_3W_0 + 8W_3W_2 + 2W_2W_0 + W_1W_0 + 3W_1W_3 +\\
   &~ 7W_3W_5 + 11W_4W_3 + 6W_4W_5 + 3W_3^2 + 4W_2W_1 + W_5W_1 +  W_4W_1 +\\
   &~  4W_4W_2 + 4W_4^2 + 2W_5^2,
\end{align*}
and the symmetrized weight enumerator of $\mathfrak{C}_2^{\perp}$ is given by
\begin{align*}
   \mathcal{W}_{\mathbf{S}}^{(\mathfrak{C}_2^{\perp})}(W_0,W_1,W_2,\ldots,W_6) =&~ W_0^2 + 2W_3W_0 + 3W_1^2 + 5W_4W_1 + 4W_2W_1 + 2W_5W_1 +  8W_4W_2 +\\
   &~ 3W_2^2 + 3W_5W_2 + 2W_1W_3 + 6W_4W_3 + 5W_3^2 +  2W_4^2 + 8W_3W_2 +\\
   &~ W_0W_2 + 4W_3W_5 + 2W_5W_4 +  W_0W_5 +  W_1W_6 + W_4W_6.
\end{align*}
\end{example}
Next, we define the Lee weight enumerator.

\begin{definition}
 Let $\mathfrak{C}$ be a $\mathbb{Z}_p\mathcal{R}\mathcal{S}$-additive code of length $n$. Then the Lee weight enumerator of $\mathfrak{C}$ is denoted by $\mathcal{W}_{\mathbf{L}}^{(\mathfrak{C})}$ and is defined by$$\mathcal{W}_{\mathbf{L}}^{(\mathfrak{C})}(x,y) = \mathcal{W}_{\mathbf{H}}^{(\overline{\phi}(\mathfrak{C}))}(x,y) = \sum_{c\in\mathfrak{C}} x^{6n-wt_L(c)} y^{wt_L(c)}$$
\end{definition}
The following result gives a relation between the Lee weight enumerator and the symmetrized weight enumerator.

\begin{theorem}\label{Theorem17}
Let $\mathfrak{C}$ be a $\mathbb{Z}_p\mathcal{R}\mathcal{S}$-additive code of length $n$. Then$$\mathcal{W}_{\mathbf{L}}^{(\mathfrak{C})}(x,y) = \mathcal{W}_{\mathbf{S}}^{(\mathfrak{C})}(x^6,x^5y,x^4y^2,x^3y^3,x^2y^4,xy^5,y^6).$$
\end{theorem}
\begin{proof}
Take $c\in\mathfrak{C}$. We have $wt_L(c) = n_1(c)+2n_2(c)+3n_3(c)+\cdots+6n_6(c)$ and $n = n_0(c)+n_1(c)+n_2(c)+\cdots+n_6(c)$. Then
\begin{align*}
    \mathcal{W}_{\mathbf{L}}^{(\mathfrak{C})}(x,y) =&~ \sum_{c\in\mathfrak{C}} x^{6n-wt_L(c)} y^{wt_L(c)}\\
    =&~ \sum_{c\in\mathfrak{C}} x^{6n_0(c)+5n_1(c)+4n_2(c)+3n_3(c)+2n_4(c)+n_5(c)}~ y^{n_1(c)+2n_2(c)+3n_3(c)+\cdots+6n_6(c)}\\
    =&~ \sum_{c\in\mathfrak{C}} (x^6)^{n_0(c)}(x^5y)^{n_1(c)}(x^4y^2)^{n_2(c)}(x^3y^3)^{n_3(c)}(x^2y^4)^{n_4(c)}(xy^5)^{n_5(c)}(y^6)^{n_6(c)}\\
    =&~ \mathcal{W}_{\mathbf{S}}^{(\mathfrak{C})}(x^6,x^5y,x^4y^2,x^3y^3,x^2y^4,xy^5,y^6) \quad\quad\quad\text{From equation } (1).
\end{align*}
\end{proof}
Using Theorem \ref{Theorem16} and Theorem \ref{Theorem17}, we find the MacWilliams identity corresponding to the Lee weight enumerator.

\begin{theorem}\label{Theorem18}
If $\mathfrak{C}$ is a $\mathbb{Z}_p\mathcal{R}\mathcal{S}$-additive code of length $n$ then$$\mathcal{W}_{\mathbf{L}}^{(\mathfrak{C}^{\perp})}(x,y) = \frac{1}{|\mathfrak{C}|}\mathcal{W}_{\mathbf{L}}^{(\mathfrak{C})}(x+y,x-y).$$
\end{theorem}
\begin{example}
Again, we consider the code $\mathfrak{C}_2$, defined in Example \ref{Example2}. Then using Theorem \ref{Theorem17} and Example \ref{Example4}, the Lee weight enumerator of $\mathfrak{C}_2$ is
\begin{align*}
    \mathcal{W}_{\mathbf{L}}^{(\mathfrak{C}_2)}(x,y) =&~ x^{12} + 5x^8y^4 + x^9y^3 + 8x^7y^5 + 2x^{10}y^2 + x^{11}y + 3x^8y^4 + 7x^4y^8 + 11x^5y^7 + 6x^3y^9 +\\
    &~ 3x^6y^6 + 4x^9y^3 + x^6y^6 + x^7y^5 + 4x^6y^6 + 4x^4y^8 + 2x^2y^{10}\\
    =&~ x^{12} + x^{11}y + 2x^{10}y^2 + 5x^9y^3 + 8x^8y^4 + 9x^7y^5 + 8x^6y^6 + 11x^5y^7 + 11x^4y^8 +\\
    &~ 6x^3y^9 + 2x^2y^{10}.
\end{align*}
Now, using the above expression and Theorem \ref{Theorem18}, the Lee weight enumerator of $\mathfrak{C}_2^{\perp}$ is given by
\begin{align*}
    \mathcal{W}_{\mathbf{L}}^{(\mathfrak{C}_2^{\perp})}(x,y) =&~ \frac{1}{64}\mathcal{W}_{\mathbf{L}}^{(\mathfrak{C}_2)}(x+y,x-y)\\
    =&~ \frac{1}{64}[(x+y)^{12} + (x+y)^{11}(x-y) + 2(x+y)^{10}(x-y)^2 + 5(x+y)^9(x-y)^3 +\\
    &~ 8(x+y)^8(x-y)^4 + 9(x+y)^7(x-y)^5 + 8(x+y)^6(x-y)^6 + 11(x+y)^5(x-y)^7 +\\
    &~ 11(x+y)^4(x-y)^8 + 6(x+y)^3(x-y)^9 + 2(x+y)^2(x-y)^{10}]\\
    =&~ x^{12} + 4x^{10}y^2 + 6x^9y^3 + 5x^8y^4 + 14x^7y^5 + 15x^6y^6 + 10x^5y^7 + 6x^4y^8 + 2x^3y^9 + x^2y^{10}.
\end{align*}
\end{example}

 \section{Quantum codes over $\mathcal{R}$ and $\mathcal{S}$}
 In this section, we find the quantum codes using CSS construction. The CSS construction uses classical linear codes to find quantum stabilizer codes. Here, we find the dual of a given classical linear code and then give conditions for dual-containing properties in order to apply CSS construction. We use magma \cite{magma} to find codes and their distances.
 \begin{theorem}
 Let $\mathcal{C}$ and $\mathcal{C}'$ be two linear codes over finite field $\mathbb{Z}_p$ with parameters $[n,k,d]_p$ and $[n,k',d']_p$, respectively. If $\mathcal{C}'\subset \mathcal{C}$, then there exists a QECC with the parameters $[[n,k+k',\mathfrak{d}]]$ where $\mathfrak{d}=min\{wt(c):c \in (\mathcal{C}\setminus \mathcal{C} ^{\prime\perp} )\cup (\mathcal{C}' \setminus \mathcal{C} ^\perp ) \}$. In particular, if $\mathcal{C}^\perp \subset \mathcal{C},$ then there exists a quantum code with the parameters $[[n,2k-n,d]]_p$.
 \end{theorem}
 For a polynomial $g(x)$ of degree $m$, reciprocal of $g(x)$ is defined by the polynomial $g^\ast= x^mg(\frac{1}{x})$. In the following, we find the dual of a cyclic code over rings $\mathcal{R}$ and $\mathcal{S}$.




We know that a cyclic code of length $s$ over the ring $\mathcal{R}$ is an ideal of $\frac{\mathcal{R}[x]}{\langle x^s-1\rangle}$. Let, $x^s-1 = f_0f_1\ldots f_{l_1}$ be a factorization of $x^s-1$ into pairwise co-prime basic irreducible factors over $\mathcal{R}$. By using the Chinese Remainder Theorem, we have
\begin{align*}
   \frac{\mathcal{R}[x]}{\langle x^s-1\rangle} =&~ \displaystyle\frac{\mathcal{R}[x]}{\bigcap_{i=0}^{l_1}\langle f_i\rangle} =~ \bigoplus_{i=0}^{l_1} \frac{\mathcal{R}[x]}{\langle f_i\rangle}
\end{align*}
and thus any ideal of $\frac{\mathcal{R}[x]}{\langle x^s-1\rangle}$ can be expressed as a sum of ideals of the rings $\frac{\mathcal{R}[x]}{\langle f_i\rangle}$. An ideal of the ring $\frac{\mathcal{R}[x]}{\langle f_i\rangle}$ is of the form $\langle u^j+\langle f_i\rangle\rangle$, where $j=0,1$. Further, the ideal $\langle u^j+\langle f_i\rangle\rangle$ is isomorphic to the ideal $\langle u^j \widehat{f_i}+\langle x^s-1\rangle\rangle$ in $\frac{\mathcal{R}[x]}{\langle x^s-1\rangle}$, where we use the notation $\widehat{f_i}$ to denote $\frac{x^s-1}{f_i}$. Now, we can obtain the form of the cyclic codes and their dual over $\mathcal{R}$.

Let $\mathcal{C}$ be a cyclic code of length $s$ over the ring $\mathcal{R}$. Then there exists pairwise co-prime polynomials $F_0, F_1$ in $\mathcal{R}[x]$ such that $$\mathcal{C}=\langle \widehat{F_0}(x),u\widehat{F_1}(x) \rangle,$$ where $F_0(x)F_1(x) = (x^s-1)$ and $\gcd(p,s)=1$. The dual code of $\mathcal{C}$ is given by$$\mathcal{C}^{\perp}=\langle \widehat{F_0}^\ast(x),u\widehat{F_1}^\ast(x) \rangle.$$ Then, by using the computer algebra system Magma, we find dual-containing codes to obtain quantum codes. Following are the examples of quantum codes using chain ring $\mathcal{R}$.

\begin{example}
Let $p=17,s=8$ and $\mathcal{C}$ be a cyclic code of length $8$ over $\mathcal{R}$. Now we have
$$x^{8}-1=(x+1)(x+2)(x+4)(x+8)(x+9)(x+13)(x+15)(x+16).$$
Take \begin{align*}
 F_0&=(x+1)(x+2)(x+4)(x+8)(x+16),\\ F_1&=(x+9)(x+13),\text{ and }\\F_2&= x+15.
\end{align*} Then generator of a code $\mathcal{C}$ is $\langle x^3 + 3x^2 +5x+4,u(x^6 +12x^5+10x^4 + 8x^3 +14x^2+ 14x + 9) \rangle$. The code $\mathcal{C}$ is a dual-containing linear code over $\mathbb{F}_{17}$ with the parameters $[16,12,4]_{17}$. By the CSS construction, we get a quantum code of the parameters $[[16,8,4]]_{17}$, which is an almost MDS code.
\end{example}

The following result gives additive quantum codes over the ring $\mathcal{R}\mathcal{S}$.

\begin{proposition}
Let $\mathcal{C}=\mathcal{C}_r\times \mathcal{C}_s$ be an $\mathcal{R}\mathcal{S}$-additive code of block length $(r,s)$ such that $\mathcal{C}$ is also separable. Then $\mathcal{C}$ satisfies dual-containing property over the ring $\mathcal{R}\mathcal{S}$ if and only if $\mathcal{C}_r $ and $\mathcal{C}_s$ satisfy dual-containing properties over the rings $\mathcal{R}$ and $\mathcal{S}$, respectively.
\end{proposition}
\begin{remark}
The codes given in Table \ref{tab:one} are optimal according to \cite{Gao}, \cite{Habibul4}, in which some codes are new, and many $[[n,k,d]]$ quantum codes satisfy $n+2-(k+2d)=2$.
\end{remark}

	\begin{table}[ht]
		\caption{Quantum codes obtained from dual containing cyclic code over $\mathcal{R}$.}
		\renewcommand{\arraystretch}{1.5}
		\begin{center}
			\begin{tabular}{|c|c|c|c|c|c|c|}
				\hline
			$p$	& $s$   & Generator polynomials & $\phi_1(C)$ &  $[[en,k, \geq  d]]_q$ \\
				\hline
			$5$	& 	$8$  & $\langle x^3 + 2x^2 + 3x + 1,u(x^6 + 2x^4 + 4x^2 + 3)\rangle$  & $[16,12,3]$ & $[[16,8,3]]_{5}$  \\
			$5$	& 	$8$  & $\langle x^2+2,u(x^6 + 3x^4 + 4x^2 + 2)\rangle$  & $[16,14,2]$ & $[[16,12,2]]_{5}$   \\
		
			$13$	& 	$6$  & $\langle x^3 + 8x^2 + 6x + 12,u(
           x^4 + 9x^3 + 12x + 4) \rangle$  & $[12,8,4]$ & $[[12,4,4]]_{13}$   \\

			$13$	& 	$8$  & $\langle x^3 + 5x^2 + 5x + 12,u(
            x^6 + 8x^4 + 12x^2 + 5) \rangle$  & $[16,12,3]$ & $[[16,8,3]]_{13}$   \\

			$13$	& 	$8$  & $\langle x^2 + 6x + 5,u(
            x^6 + 7x^5 + 5x^4 + x^2 + 7x + 5) \rangle$  & $[16,14,2]$ & $[[16,12,2]]_{13}$   \\

			$13$	& 	$18$  & $\langle x^8 + 10x^6 + 7x^5 + 5x^3 + 12x^2 + 3,u(x^{18}-1)/(x^4 + 4x^3 + 3x + 12) \rangle$  & $[36,24,5]$ & $[[36,12,5]]_{13}$   \\
			
	    	$17$	& 	$8$  & $\langle x^3 + 3x^2 +5x+4,u(x^6 +12x^5+10x^4 + 8x^3 +14x^2+ 14x + 9) \rangle$  & $[16,12,4]$ & $[[16,8,4]]_{17}$   \\

				\hline
			\end{tabular}\label{tab:one}
		\end{center}
	\end{table}

\section{Conclusion}
We have considered the Frobenius rings $\mathcal{R},~\mathcal{S},~\mathcal{R}\mathcal{S},~\text{and }\mathbb{Z}_p\mathcal{R}\mathcal{S}$, and studied the additive constacyclic codes over these rings. By defining suitable inner products, we have determined the generators of the constacyclic codes and their duals. Also, we have defined Gray maps on $\mathcal{R},~\mathcal{S},~\text{and }\mathbb{Z}_p\mathcal{R}\mathcal{S}$, and studied the images under these maps. We have derived a few results on the Gray images of additive cyclic and additive constacyclic codes. Further, we have defined several weight enumerators, such as the complete weight enumerator, the Hamming weight enumerator, the symmetrized weight enumerator and the Lee weight enumerator, and obtained the MacWilliams identities corresponding to each of these weight enumerators. Finally, applying the CSS construction, we have obtained some quantum codes over $\mathcal{R}$.
\section*{Acknowledgement}
The first and second authors are thankful to the Council of Scientific \& Industrial Research (CSIR) (under grant no. 09/1023(0027)/2019-EMR-I and 09/1023(0030)/2019-EMR-I) for financial support and the Indian Institute of Technology Patna for providing research facilities. The authors would also like to thank the Editor and anonymous referee(s) for providing constructive suggestions to improve the presentation of the manuscript.
\section*{Declarations}
\textbf{Data Availability Statement}: The authors declare that [the/all other] data supporting the findings of this study are available within the article. \\
\textbf{Competing interests}: The authors declare that there is no conflict of interest regarding the publication of this manuscript.\\

\end{document}